\documentclass[sigconf, nonacm]{acmart}

\usepackage{libertine}

\usepackage{datetime}
\usepackage{url}
\usepackage{paralist}

\usepackage[font=small,labelfont=bf,justification=centering]{caption}

\usepackage{stmaryrd}

\usepackage{amsmath,amsthm}
\newtheorem{theorem}{Theorem}

\theoremstyle{definition}
\newtheorem{lemma}{Lemma}

\theoremstyle{definition}
\newtheorem{definition}{Definition}

\theoremstyle{definition}

\theoremstyle{definition}
\newtheorem{corollary}{Corollary}

\newenvironment{sketch}{\textsc{Proof Sketch.}}{\hfill$\square$\smallskip}

\theoremstyle{definition}

\theoremstyle{definition}
\newtheorem{assumption}{Assumption}

\usepackage{cleveref}
\crefname{section}{\S}{\S\S}
\Crefname{section}{\S}{\S\S}
\crefformat{section}{\S#2#1#3}
\Crefname{line}{Line}{line}
\crefname{line}{Line}{line}
\Crefname{assumption}{Assumption}{assumption}

\crefname{table}{Table}{Tables}

\usepackage{graphicx}
\usepackage{subfig}
\usepackage{arydshln}

\usepackage{multirow}

\usepackage{threeparttable}

\usepackage{algorithmicx}
\usepackage{algorithm} 
\usepackage{algpseudocode}
\usepackage{algcompatible}

\usepackage{xspace}

\newcommand{\allprocesses}{\Pi}
\newcommand{\allconfigurations}{\mathcal{I}}

\newcommand{\general}{\textsc{Universal}\xspace}

\newcommand{\process}[1]{\pi(#1)}
\newcommand{\processes}[1]{\pi(#1)}

\newcommand{\similar}[1]{\mathit{sim}(#1)}

\usepackage{xcolor}
\usepackage{framed}
\definecolor{lightgray}{gray}{0.90}
\renewenvironment{leftbar}[1][\hsize]
{%
\MakeFramed{\hsize#1\advance\hsize-\width\FrameRestore}%
}
{\endMakeFramed}

\algsetblockdefx[UponSend]{UponSend}{EndUponSend}{}{}[1]{\textbf{upon} \textsc{send}($#1$)}{\phantom}
\algtext*{EndUponSend}

\algsetblockdefx[UponDeliver]{UponDeliver}{EndUponDeliver}{}{}[1]{\textbf{upon} \textsc{SecureDeliver}($#1$)}{\phantom}
\algtext*{EndUponDeliver}

\algsetblockdefx[UponInit]{UponInit}{EndUponInit}{}{}[0]{\textbf{upon} \textsc{init}}{\phantom}
\algtext*{EndUponInit}

\algsetblockdefx[UponExists]{UponExists}{EndUponExists}{}{}[2]{\textbf{upon exists} $#1$ such that $#2$}{\phantom}
\algtext*{EndUponExists}

\algsetblockdefx[Procedure]{Procedure}{EndProcedure}{}{}[2]{\textbf{procedure} {#1}($#2$)}{\phantom}
\algtext*{EndProcedure}
\algsetblockdefx[Function]{Function}{EndFunction}{}{}[2]{\textbf{function} {#1}($#2$)}{\phantom}
\algtext*{EndFunction}

\algnewcommand{\BlueComment}[1]{\textcolor{blue}{\hfill\(\triangleright\) #1}}
\algnewcommand{\LineComment}[1]{\State \(\triangleright\) #1}

\usepackage{listings}
\crefname{lstlisting}{listing}{listings}
\Crefname{lstlisting}{Listing}{Listings}

\crefname{code}{line}{lines}
\Crefname{code}{Line}{Lines}

\usepackage{xcolor}

\definecolor{mygreen}{rgb}{0.254,0.572,0.294}
\definecolor{mygray}{rgb}{0.5,0.5,0.5}
\definecolor{myorange}{rgb}{1,0.35,0}
\definecolor{mymauve}{rgb}{0.58,0,0.82}
\definecolor{myblue}{rgb}{0.2,0.4,0.6}

\definecolor{rakos4orange}{RGB}{255,165,0}
\definecolor{rakos4blue}{RGB}{14,48,173}
\definecolor{rakos4lblue}{RGB}{92,172,238}
\definecolor{rakos4dgray}{RGB}{77,77,77}
\definecolor{plainred}{RGB}{211,63,63}
\definecolor{plainorange}{RGB}{221,105,41}

\usepackage{bold-extra} 

\lstdefinelanguage{Golang}%
  {morekeywords=[1]{package,import,struct,defer,panic,%
     recover,select,var,const,iota,},%
   morekeywords=[2]{string,uint,uint8,uint16,uint32,uint64,int,int8,int16,%
     int32,int64,bool,float32,float64,complex64,complex128,byte,rune,uintptr,%
     error,interface,message,node},%
   morekeywords=[3]{map,slice,make,new,nil,len,cap,copy,close,true,false,%
     delete,append,real,imag,complex,chan,},%
   morekeywords=[4]{break,continue,goto,switch,case,fallthrough,%
    default,},%
   morekeywords=[5]{Println,Printf,Error,Send},%
   sensitive=true,%
   morecomment=[l]{//},%
   morecomment=[s]{/*}{*/},%
   morestring=[b]",%
   morestring=[s]{`}{`},%
   }

\usepackage{thmtools,thm-restate}

\usepackage{algorithmicx}
\usepackage{algorithm} 
\usepackage{algpseudocode}
\usepackage{xifthen}
\usepackage{graphicx}
\usepackage{wrapfig}
\usepackage{dblfloatfix}

\usepackage{stmaryrd}
\usepackage{amsmath}
\usepackage{amsthm}

\usepackage{cancel}
\newif\ifcomments
\commentstrue

\usepackage[normalem]{ulem}
\newcommand{\ms}[1]{%
	    \relax\ifmmode
	        \mathord{\mathcode`\-="702D\it #1\mathcode`\-="2200}%
	    \else
	        {\it #1}%
	    \fi
}
\newcommand{\tup}[1]{%
	    \relax\ifmmode
	      \langle #1 \rangle%
	    \else
	        $\langle$ #1 $\rangle$%
	    \fi
}

\date{}

\setcopyright{none} 
\acmConference[PODC 2023]{PODC 2023}{TBD}{Orlando, USA; TBD}
\acmYear{2019}



\begin{document}

\author{Pierre Civit}
\affiliation{
\institution{Sorbonne University} 
\country{France}
}

\author{Seth Gilbert}
\affiliation{
\institution{NUS Singapore}
\country{Singapore}
}

\author{Rachid Guerraoui}
\affiliation{
\institution{École Polytechnique Fédérale de Lausanne (EPFL)}
\country{Switzerland}
}

\author{Jovan Komatovic}
\affiliation{
\institution{École Polytechnique Fédérale de Lausanne (EPFL)}
\country{Switzerland}
}

\author{Manuel Vidigueira}
\affiliation{
\institution{École Polytechnique Fédérale de Lausanne (EPFL)}
\country{Switzerland}
}

\title{On the Validity of Consensus \footnotesize{(Extended Version)}}

\begin{abstract}
The Byzantine consensus problem involves $n$ processes, out of which $t < n$ 
could be faulty and behave arbitrarily. 
Three properties characterize consensus: 
(1) termination, requiring correct (non-faulty)  processes to eventually reach a decision, 
(2) agreement, preventing them from deciding different values, and 
(3) validity, precluding ``unreasonable'' decisions.
But, what is a reasonable decision?
Strong validity, a classical property, stipulates that, if all correct processes propose the same value, only that value can be decided.
Weak validity, another established property, stipulates that, if all processes are correct and they propose the same value, that value must be decided.
The space of possible validity properties is vast. 
Yet, their impact on consensus algorithms remains unclear.

This paper addresses the question of which validity properties allow Byzantine consensus to be solvable in a general partially synchronous model, and at what cost.
First, we determine the necessary and sufficient conditions for a validity property to make the consensus problem solvable; we say that such validity properties are \emph{solvable}.
Notably, we prove that, if $n \leq 3t$, all solvable validity properties are \emph{trivial} (there exists an always-admissible decision).
Furthermore, we show that, with any non-trivial (and solvable) validity property, consensus requires $\Omega(t^2)$ messages. 
This extends the seminal Dolev-Reischuk bound, originally proven for strong validity, to \emph{all} non-trivial validity properties. 
Lastly,  we give a Byzantine consensus algorithm, we call \general, for \emph{any} solvable (and non-trivial) validity property.
Importantly, \general incurs $O(n^2)$ message complexity.
Thus, together with our lower bound, \general implies a fundamental result in partial synchrony: with $t \in \Omega(n)$, the message complexity of all (non-trivial) consensus variants is $\Theta(n^2)$.
\end{abstract}

\maketitle

\section{Introduction}

Consensus~\cite{LSP82} is the cornerstone of state machine replication (SMR)~\cite{CL02,adya2002farsite,abd2005fault,kotla2004high,veronese2011efficient,amir2006scaling,kotla2007zyzzyva,malkhi2019flexible,momose2021multi}, as well as various distributed algorithms~\cite{GS01, ben2019completeness,galil1987cryptographic, gilbert2010rambo}.
Recently, it has received a lot of attention with the advent of blockchain systems~\cite{abraham2016solida,chen2016algorand,abraham2016solidus,luu2015scp,correia2019byzantine,CGL18,buchman2016tendermint}.
The consensus problem is posed in a system of $n$ processes, out of which $t < n$ can be \emph{faulty}, and the rest \emph{correct}.
Each correct process proposes a value, and consensus enables correct processes to decide on a common value.
In this paper, we consider Byzantine~\cite{LSP82} consensus, where faulty processes can behave arbitrarily.
While the exact definition of the problem might vary, two properties are always present: (1) \emph{termination}, requiring correct  processes to eventually decide, and (2) \emph{agreement}, preventing them from deciding different values.
It is not hard to devise an algorithm that satisfies only these two properties: 
every correct process decides the same, predetermined value. 
However, this algorithm is vacuous. 
To preclude such trivial solutions and render consensus meaningful, 
an additional property is required -- \emph{validity} -- defining which decisions are admissible.

\paragraph{The many faces of validity.}
The literature contains many flavors of validity~\cite{AbrahamAD04,civit2022byzantine,CGL18,abraham2017brief,KMM03,CGG21,yin2019hotstuff,fitzi2003efficient,siu1998reaching,stolz2016byzantine,melnyk2018byzantine}.
One of the most studied properties is \emph{Strong Validity}~\cite{civit2022byzantine,CGL18,KMM03,abraham2017brief}, stipulating that, if all correct processes propose the same value, only that value can be decided.
Another common property is \emph{Weak Validity}~\cite{civit2022byzantine,CGG21,yin2019hotstuff}, affirming that, if all processes are correct and propose the same value, that value must be decided. 
While validity may appear as an inconspicuous property, its exact definition has a big impact on our understanding of consensus algorithms.
For example, the seminal Dolev-Reischuk bound~\cite{dolev1985bounds} states that any solution to consensus with \emph{Strong Validity} incurs a quadratic number of messages;  it was recently proven that the bound is tight~\cite{Momose2021,civit2022byzantine,lewis2022quadratic}.
In contrast, while there have been several improvements to the performance of consensus with \emph{Weak Validity} over the last 40 years~\cite{yin2019hotstuff, lewis2022quadratic, civit2022byzantine}, the (tight) lower bound on message complexity remains unknown.
(Although the bound is conjectured to be the same as for \emph{Strong Validity}, this has yet to be formally proven.)
Many other fundamental questions remain unanswered:
\begin{compactitem}
\item What does it take for a specific validity property to make consensus solvable? 

\item What are the (best) upper and lower bounds on the message complexity of consensus with any specific validity property?

\item Is there a hierarchy of validity properties (e.g., a ``strongest'' validity property)?
\end{compactitem}
To the best of our knowledge, no in-depth study of the validity property has ever been conducted, despite its importance and the emerging interest from the research community~\cite{aboutvalidity,CKN21}.


\paragraph{Contributions.}
We propose a precise mathematical formalism for the analysis of validity properties.
We define a validity property as a mapping from assignments of proposals into admissible decisions.
Although simple, our formalism enables us to determine the exact impact of validity on the solvability and complexity of consensus in the classical partially synchronous model~\cite{DLS88}, and answer the aforementioned open questions.
Namely, we provide the following contributions:
\begin{compactitem} 
    \item We classify all validity properties into solvable and unsolvable ones.
    (If a validity property makes consensus solvable, we say that the property itself is \emph{solvable}.)
    Specifically, for $n \leq 3t$, we show that only \emph{trivial} validity properties (for which there exists an always-admissible decision) are solvable.
    In the case of $n > 3t$, we define the \emph{similarity condition}, which we prove to be necessary and sufficient for a validity property to be solvable.
    
    \item
    We prove that all non-trivial (and solvable) validity properties require $\Omega(t^2)$ exchanged messages.
    This result extends the Dolev-Reischuk bound~\cite{dolev1985bounds}, proven only for \emph{Strong Validity}, to \emph{all} ``reasonable'' validity properties.
    
    \item 
    Finally, we present \general, a consensus algorithm for all solvable (and non-trivial) validity properties.
    Assuming a public-key infrastructure, \general exchanges $O(n^2)$ messages.
    Thus, together with our lower bound, \general implies a fundamental result in partial synchrony: given $t \in \Omega(n)$, all (non-trivial) consensus variants have $\Theta(n^2)$ message complexity.
    \Cref{fig:classification_landscape} summarizes our findings.
\end{compactitem}

\begin{figure}[ht]
    \centering
    \includegraphics[scale = 0.45]{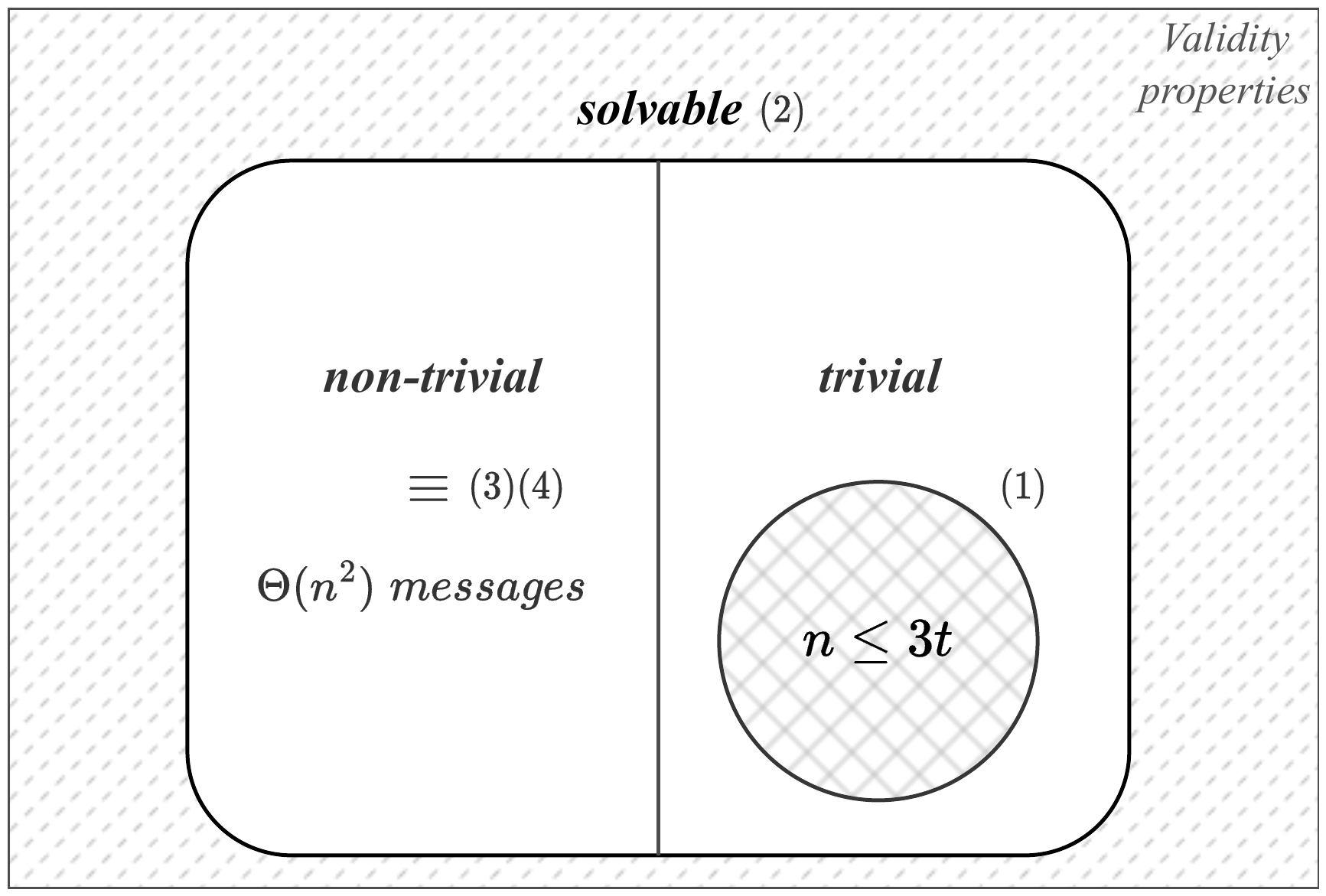}
    \caption{Illustration of our results: (1) with $n \leq 3t$, all solvable validity properties are trivial; (2) the exact set of solvable validity properties (as determined by our necessary and sufficient conditions); (3) all non-trivial (and solvable) validity properties require $\Omega(t^2)$ exchanged messages; (4) for any non-trivial (and solvable) validity property, there exists a consensus algorithm with $O(n^2)$ message complexity.}
    \label{fig:classification_landscape}
\end{figure}

\paragraph{Technical overview.}
In our formalism, we use the notion of an \emph{input configuration} that denotes an 
assignment of proposals to correct processes.
For example, $\big[ (P_1, v), (P_2, v), (P_3, v) \big]$ represents an input configuration by which (1) only processes $P_1$, $P_2$, and $P_3$ are correct, and (2) processes $P_1$, $P_2$, and $P_3$ propose $v$.

First, we define a similarity relation between input configurations: two input configurations are similar if and only if (1) they have (at least) one process in common, and (2) for every common process, the process's proposal is identical in both input configurations.
For example, an input configuration $c = \big[ (P_1, 0), (P_2, 1) \big]$ is similar to $c' = \big[ (P_1, 0), (P_3, 0) \big]$, but not to $\big[ (P_1, 0), (P_2, 0) \big]$.
We observe that all similar input configurations must have an admissible value in common; we call this \emph{canonical similarity}.
Let us illustrate why a common admissible value must exist.
Consider the aforementioned similar input configurations $c = \big[ (P_1, 0), (P_2, 1) \big]$ and $c' = \big[ (P_1, 0), (P_3, 0) \big]$.
If there is no common admissible value for $c$ and $c'$, consensus cannot be solved: process $P_1$ cannot distinguish (1) an execution in which $P_2$ is correct, and $P_3$ is faulty and silent, from (2) an execution in which $P_2$ is faulty, but behaves correctly, and $P_3$ is correct, but slow.
Thus, $P_1$ cannot conclude whether it needs to decide an admissible value for $c$ or for $c'$.
Canonical similarity is a critical intermediate result that we use extensively throughout the paper (even if it does not directly imply any of our results).

In our proof of triviality with $n \leq 3t$, we intertwine the classical partitioning argument~\cite{lamport2019byzantine} with our canonical similarity result.
Namely, we show that, for any input configuration, there exists an execution in which the same value $x$ is decided, making $x$ an always-admissible value. 
For our lower bound, while following the idea of the original proof~\cite{dolev1985bounds,DBLP:conf/podc/AbrahamCDNP0S19}, we rely on canonical similarity to prove the bound for all solvable and non-trivial validity properties.
Finally, we design \general by relying on vector consensus~\cite{neves2005solving,vaidya2013byzantine,Doudou1998,correia2006consensus}, a problem in which processes agree on the proposals of $n - t$ processes: when a correct process decides a vector $\mathit{vec}$ of $n - t$ proposals (from vector consensus), it decides from \general the common admissible value for all input configurations similar to $\mathit{vec}$.
For example, consider an execution which corresponds to an input configuration $c$.
If, in this execution, a correct process decides a vector $\mathit{vec}$ from vector consensus, it is guaranteed that $\mathit{vec}$ is similar to $c$ (the proposals of correct processes are identical in $c$ and in $\mathit{vec}$).
Hence, deciding (from \general) the common admissible value for all input configurations similar to $\mathit{vec}$ guarantees that the decided value is admissible according to $c$.



\paragraph{Roadmap.}
We provide an overview of related work in \Cref{section:related_work}.
In \Cref{section:preliminaries}, we specify the system model (\Cref{subsection:computational_model_main}), define the consensus problem (\Cref{subsection:byzantine_consensus}), describe our formalism for validity properties (\Cref{subsection:specifying_validity}), and present canonical similarity (\Cref{subsection:similarity}).
We define the necessary conditions for the solvability of validity properties in \Cref{section:classification}.
In \Cref{section:lower_upper_bounds}, we prove a quadratic lower bound on message complexity for all non-trivial (and solvable) validity properties (\Cref{subsection:lower_bound}), and introduce \general, a general consensus algorithm for any solvable (and non-trivial) validity property (\Cref{subsection:general_algorithm}).
We conclude the paper in \Cref{section:conclusion}.
The appendix contains (1) omitted proofs and algorithms, and (2) a proposal for how to extend our formalism to accommodate for blockchain-specific validity properties.
\section{Related work}
\label{section:related_work}

\paragraph{Solvability of consensus.}
The consensus problem has been thoroughly investigated in a variety of system settings and failure models.
It has been known (for long) that consensus can be solved in a synchronous setting, both with crash~\cite{lynch1996distributed,cachin2011introduction,raynal2002consensus} and arbitrary failures~\cite{abraham2017brief,kowalski2013synchronous,schmid2004synchronous,raynal2002consensus,Momose2021}.
In an asynchronous environment, however, consensus cannot be solved deterministically even if a single process can fail, and it does so only by crashing; this is the seminal FLP impossibility result~\cite{fischer1985impossibility}.

A traditional way of circumventing the FLP impossibility result is \emph{randomization}~\cite{B83,aspnes2003randomized,abraham2019asymptotically,LL0W20,ezhilchelvan2001randomized}, where termination of consensus is not ensured deterministically.
Another well-established approach to bypass the FLP impossibility is to strengthen the communication model with \emph{partial synchrony}~\cite{DLS88}: the communication is asynchronous until some unknown time, and then it becomes synchronous.
The last couple of decades have produced many partially synchronous consensus algorithms~\cite{yin2019hotstuff,CGL18,CGG21,CL02,DLS88,BKM19,lynch1996distributed,lamport2001paxos,civit2022byzantine,lewis2022quadratic}.

Another line of research has consisted in weakening the definition of  consensus  to make it deterministically solvable under asynchrony.
In the \emph{condition-based} approach~\cite{mostefaoui2003conditions}, 
the specification of consensus is relaxed to require termination only if the assignment of proposals satisfies some predetermined conditions.
The efficiency of this elegant approach has been studied further in~\cite{mostefaoui2001hierarchy}.
Moreover, the approach is extended to the synchronous setting~\cite{mostefaoui2003using,zibin2003condition}, as well as  to the $k$-set agreement problem~\cite{hagit2002wait,mostefaoui2002condition}.

\paragraph{Solvability of general decision problems.}
A distributed decision problem has been defined in~\cite{crime_punishment,MendesTH14,HerlihyS99} 
as a mapping from input assignments to admissible decisions. 
Our validity formalism is of the same nature, and it is inspired by the aforementioned specification of decision problems.

The solvability of decision problems has been thoroughly studied in asynchronous, crash-prone settings.
It was shown in~\cite{moran1987extended} that the FLP impossibility result~\cite{fischer1985impossibility} can be extended to many decision problems.
In~\cite{biran1990combinatorial}, the authors defined the necessary and sufficient conditions for a decision problem to be asynchronously solvable with a single crash failure.
The asynchronous solvability of problems in which crash failures occur at the very beginning of an execution was studied in~\cite{taubenfeld1989initial}.
The necessary and sufficient conditions for a decision problem to be asynchronously solvable (assuming a crash-prone setting) in a randomized manner were given in ~\cite{Chor1989}. 
The topology-based approach on studying the solvability of decision problems in asynchrony has proven to be extremely effective, both for crash~\cite{HerlihyS93,SaraphHG18,Herlihy2013} and arbitrary failures~\cite{MendesTH14,Herlihy2013}.
Our results follow the same spirit as many of these approaches; however, we study the deterministic solvability and complexity of all consensus variants in a partially synchronous environment.

\paragraph{Validity of consensus.}
Various validity properties have been  associated with the consensus problem (beyond the aforementioned \emph{Strong Validity} and \emph{Weak Validity}).
\emph{Correct-Proposal Validity}~\cite{fitzi2003efficient,siu1998reaching} states that  a value decided by a correct process must have been proposed by a correct process. 
\emph{Median Validity}~\cite{stolz2016byzantine} is a validity property proposed in the context of synchronous consensus,  requiring the decision to be close to the median of the proposals of correct processes. \emph{Interval Validity}~\cite{melnyk2018byzantine}, on the other hand, requires the  decision to be close to the $k$-th smallest proposal of correct processes.
The advent of blockchain technologies has resurged the concept of \emph{External Validity}~\cite{Cachin2001,BKM19,yin2019hotstuff}.
This property requires the decided value to satisfy a predetermined predicate, typically asserting whether the decided value follows the rules of a blockchain system (e.g., no double-spending).
(This paper considers a simple formalism to express basic validity properties and derive our results. To express \emph{External Validity}, which is out of the scope of the paper, we propose an incomplete extension of our formalism in \Cref{section:extended_formalim_appendix}, and leave its realization for future work.)
\emph{Convex-Hull Validity}, which states that the decision must belong to the convex hull of the proposals made by the correct processes, is employed in approximate agreement~\cite{AbrahamAD04,MendesH13,GhineaLW22,ghinea2023multidimensional}.
We underline that the approximate agreement problem is not covered in this paper since the problem allows correct processes to disagree as long as their decisions are ``close'' to each other.
In our paper, we do study (in a general manner) the utilization of \emph{Convex-Hull Validity} in the classical consensus problem (in which the correct processes are required to ``exactly'' agree).

In interactive consistency~\cite{maneas2021achieving,fischer1981lower,ben2003resilient},  correct processes agree on the proposals of all correct processes. Given that the problem is  impossible in a non-synchronous setting, a weaker variant has been considered: vector consensus~\cite{neves2005solving,vaidya2013byzantine,Doudou1998,correia2006consensus,duan2023practical,Ben-Or94}. Here, processes need to agree on a vector of proposals which does not necessarily include the proposals of \emph{all} correct processes. Interactive consistency and vector consensus can be seen as specific consensus problems with a validity property requiring that, if a decided vector contains a proposal $v$ of a correct process, that correct process has indeed proposed $v$.
The design of \general, our general consensus algorithm for any solvable (and non-trivial) validity property, demonstrates that any non-trivial flavor of consensus which is solvable in partial synchrony can be solved using vector consensus (see \Cref{subsection:general_algorithm}).

\section{Preliminaries} \label{section:preliminaries}

In this section, we present the computational model (\Cref{subsection:computational_model_main}), recall the  consensus problem (\Cref{subsection:byzantine_consensus}), formally define validity properties (\Cref{subsection:specifying_validity}), and introduce canonical similarity (\Cref{subsection:similarity}).

\subsection{Computational Model} \label{subsection:computational_model_main}

\paragraph{Processes.}
We consider a system $\allprocesses = \{P_1, P_2, ..., P_n\}$ of $n$ processes; each process is a deterministic state machine.
At most $t$ ($0 < t < n$) processes can be \emph{faulty}: those processes can exhibit arbitrary behavior.
A non-faulty process is said to be \emph{correct}.
Processes communicate by exchanging messages over an authenticated point-to-point network.
The communication network is \emph{reliable}: if a correct process sends a message to a correct process, the message is eventually received.
Each process has its own local clock, and no process can take infinitely many computational steps in finite time.

\paragraph{Executions.}
Given an algorithm $\mathcal{A}$, $\mathit{execs}(\mathcal{A})$ denotes the set of all executions of $\mathcal{A}$.
Furthermore, $\mathit{Corr}_{\mathcal{A}}(\mathcal{E})$ denotes the set of correct processes in an execution $\mathcal{E} \in \mathit{execs}(\mathcal{A})$.
We say that an execution $\mathcal{E} \in \mathit{execs}(\mathcal{A})$ is \emph{canonical} if and only if no faulty process takes any computational step in $\mathcal{E}$; note that faulty processes do not send any message in a canonical execution.
Moreover, observe that any execution $\mathcal{E}$ with $\mathit{Corr}_{\mathcal{A}}(\mathcal{E}) = \allprocesses$ is canonical.

\paragraph{Partial synchrony.}
We consider the standard partially synchronous model~\cite{DLS88}. 
For every execution of the system, there exists a Global Stabilization Time (GST) and a positive duration $\delta$ such that message delays are bounded by $\delta$ after GST.
GST is not known to the processes, whereas $\delta$ is.
We assume that all correct processes start executing their local algorithm before or at GST.
Local clocks may drift arbitrarily before GST, but do not drift thereafter.

We remark that (almost) all results presented in the paper hold even if $\delta$ is unknown.
Namely, the classification of all validity properties remains the same: if a validity property is solvable (resp., unsolvable) with known $\delta$, then the validity property is solvable (resp., unsolvable) with unknown $\delta$.
Moreover, the quadratic lower bound on the message complexity trivially extends to the model in which $\delta$ is unknown.
However, the question of whether the bound is tight when $\delta$ is unknown does remain open.

\paragraph{Cryptographic primitives.}
In one variant of the \general algorithm, we assume a public-key infrastructure (PKI).
In fact, this variant relies on a closed-box consensus algorithm which internally utilizes a PKI.
In a PKI, every process knows the public key of every other process, and, when needed, processes sign messages using digital signatures.
We denote by $\langle m \rangle_{\sigma_i}$ a message $m$ signed by the process $P_i$.
Crucially, faulty processes cannot forge signatures of correct processes.

\paragraph{Message complexity.}
Let $\mathcal{A}$ be any algorithm and let $\mathcal{E} \in \mathit{execs}(\mathcal{A})$ be any execution of $\mathcal{A}$.
The message complexity of $\mathcal{E}$ is the number of messages sent by correct processes during $[\text{GST}, \infty)$.

The \emph{message complexity} of $\mathcal{A}$ is defined as
\begin{equation*}
    \max_{\mathcal{E} \in \mathit{execs}(\mathcal{A})}\bigg\{\text{message complexity of } \mathcal{E} \bigg\}.
\end{equation*}


\subsection{Consensus} \label{subsection:byzantine_consensus}

We denote by $\mathcal{V}_I$ the (potentially infinite) set of values processes can propose, and by $\mathcal{V}_O$ the (potentially infinite) set of values processes can decide.
Consensus\footnote{In the paper, we use ``consensus'' and ``Byzantine consensus'' interchangeably.} exposes the following interface:
\begin{compactitem}
    \item \textbf{request} $\mathsf{propose}(v \in \mathcal{V}_I)$: a process proposes a value $v$.
    
    \item \textbf{indication} $\mathsf{decide}(v' \in \mathcal{V}_O)$: a process decides a value $v'$.
\end{compactitem}
A correct process proposes exactly once, and it decides at most once.
Consensus requires the following properties:
\begin{compactitem}
    \item \emph{Termination:} Every correct process eventually decides.
    
    \item \emph{Agreement:} No two correct processes decide different values.
\end{compactitem}
If the  consensus problem was completely defined by \emph{Termination} and \emph{Agreement}, a trivial solution would exist: processes decide on a default value.
Therefore, the specification of consensus additionally includes a validity property, which connects the proposals of correct processes to admissible decisions, precluding the aforementioned trivial solutions.

\subsection{Validity} \label{subsection:specifying_validity}
In a nutshell, our specification of a validity property includes a set of assignments of proposals to correct processes, and, for each such assignment, a corresponding set of admissible decisions.

We start by defining a \emph{process-proposal} pair as a pair $(P, v)$, where (1) $P \in \allprocesses$ is a process, and (2) $v \in \mathcal{V}_I$ is a proposal.
Given a process-proposal pair $\mathit{pp} = (P, v)$, $\mathsf{proposal}(\mathit{pp}) = v$ denotes the proposal associated with $\mathit{pp}$.

An \emph{input configuration} is a tuple $\big[ \mathit{pp}_1, \mathit{pp}_2, ..., \mathit{pp}_x \big]$ of $x$ process-proposal pairs, where (1) $n - t \leq x \leq n$, and (2) every process-proposal pair is associated with a distinct process.
Intuitively, an input configuration represents an assignment of proposals to correct processes.
For example, $\big[ (P_1, v), (P_2, v), (P_3, v), (P_4, v), (P_5, v) \big]$ is an input configuration describing an execution in which (1) only processes $P_1$, $P_2$, $P_3$, $P_4$ and $P_5$ are correct, and (2) all of them propose the same value $v$.

We denote by $\mathcal{I}$ the set of all input configurations.
Furthermore, for every $x \in [n - t, n]$, $\mathcal{I}_x \subset \mathcal{I}$ denotes the set of input configurations with \emph{exactly} $x$ process-proposal pairs.
For every input configuration $c \in \mathcal{I}$, we denote by $c[i]$ the process-proposal pair associated with process $P_i$; if such a process-proposal pair does not exist, $c[i] = \bot$.
Lastly, $\process{c} = \{ P_i \in \Pi \,|\, c[i] \neq \bot\}$ denotes the set of all processes included in $c$.

Given (1) an execution $\mathcal{E}$ of an algorithm $\mathcal{A}$, where $\mathcal{A}$ exposes the $\mathsf{propose}(\cdot)/\mathsf{decide}(\cdot)$ interface, and (2) an input configuration $c \in \mathcal{I}$, we say that $\mathcal{E}$ \emph{corresponds} to $c$ ($\mathsf{input\_conf}(\mathcal{E}) = c$) if and only if (1) $\process{c} = \mathit{Corr}_{\mathcal{A}}(\mathcal{E})$, and (2) for every process $P_i \in \processes{c}$, $P_i$'s proposal in $\mathcal{E}$ is $\mathsf{proposal}(c[i])$.

Finally, we define a validity property $\mathit{val}$ as a function $\mathit{val}: \allconfigurations \to 2^{\mathcal{V}_O}$ such that, for every input configuration $c \in \allconfigurations$, $\mathit{val}(\mathit{c}) \neq \emptyset$.
An algorithm $\mathcal{A}$, where $\mathcal{A}$ exposes the $\mathsf{propose}(\cdot)/\mathsf{decide}(\cdot)$ interface, \emph{satisfies} a validity property $\mathit{val}$ if and only if, in any execution $\mathcal{E} \in \mathit{execs}(\mathcal{A})$, no correct process decides a value $v' \notin \mathit{val}\big( \mathsf{input\_conf}(\mathcal{E}) \big)$.
That is, an algorithm satisfies a validity property if and only if correct processes decide only admissible values.

\paragraph{Traditional properties in our formalism.}
To illustrate the expressiveness of our formalism, we describe how it can be used for \emph{Strong Validity}, \emph{Weak Validity} and \emph{Correct-Proposal Validity}.
For all of them, $\mathcal{V}_I = \mathcal{V}_O$.
\emph{Weak Validity} can be expressed as
\[
    \mathit{val}(c) =\begin{cases}
            \{v\}, & \text{if } (\processes{c} = \allprocesses) \land (\forall P_i \in \processes{c}: \mathsf{proposal}(c[i]) = v) \\
            \mathcal{V}_O, & \text{otherwise,}
    \end{cases}
\]

whereas \emph{Strong Validity} can be expressed as
\[
    \mathit{val}(c) =\begin{cases}
            \{v\}, & \text{if } \forall P_i \in \processes{c}: \mathsf{proposal}(c[i]) = v \\
            \mathcal{V}_O, & \text{otherwise.}
    \end{cases}
\]
Finally, \emph{Correct-Proposal Validity} can be expressed as
\[
    \mathit{val}(c) = \{v \,|\, \exists P_i \in \processes{c}: \mathsf{proposal}(c[i]) = v\}\text{.}
\]


\paragraph{Consensus algorithms.}
An algorithm $\mathcal{A}$ solves consensus with a validity property $\mathit{val}$ if and only if the following holds:
\begin{compactitem}
    \item $\mathcal{A}$ exposes the $\mathsf{propose}(\cdot)/\mathsf{decide}(\cdot)$ interface, and
    
    \item $\mathcal{A}$ satisfies \emph{Termination}, \emph{Agreement} and $\mathit{val}$.
\end{compactitem}

Lastly, we define the notion of a solvable validity property.

\begin{definition} [Solvable validity property]
We say that a validity property $\mathit{val}$ is \emph{solvable} if and only if there exists an algorithm which solves  consensus 
with $\mathit{val}$.
\end{definition} 

\subsection{Canonical Similarity} \label{subsection:similarity}

In this subsection, we introduce \emph{canonical similarity}, a crucial intermediate result.
In order to do so, we first establish an important relation between input configurations, that of \emph{similarity}.

\paragraph{Similarity.}

We define the similarity relation (``$\sim$'') between input configurations:
\begin{equation} \nonumber
    \begin{split}
        &\forall c_1, c_2 \in \allconfigurations: c_1 \sim c_2 \iff \\ 
        (\processes{c_1} \cap \processes{c_2} \neq~&\emptyset) \land (\forall P_i \in \processes{c_1} \cap \processes{c_2}: c_1[i] = c_2[i]).
    \end{split}
\end{equation}

In other words, $c_1$ is similar to $c_2$ if and only if (1) $c_1$ and $c_2$ have (at least) one process in common, and (2) for every common process, the process's proposal is identical in both input configurations.
For example, when $n = 3$ and $t = 1$, $c = \big[ (P_1, 0), (P_2, 1), (P_3, 0) \big]$ is similar to $\big[ (P_1, 0), (P_3, 0) \big]$, whereas $c$ is not similar to $\big[ (P_1, 0), (P_2, 0) \big]$.
Note that the similarity relation is symmetric (for every pair $c_1, c_2 \in \mathcal{I}$, $c_1 \sim c_2 \Leftrightarrow c_2 \sim c_1$) and reflexive (for every $c \in \mathcal{I}$, $c \sim c$).

For every input configuration $c \in \mathcal{I}$, we define $\mathit{sim}(c)$:
\begin{equation*}
    \mathit{sim}(c) = \{ c' \in \mathcal{I} \,|\, c' \sim c \}.
\end{equation*}



\paragraph{The result.}
Let $\mathcal{A}$ be an algorithm which solves  consensus  with some validity property $\mathit{val}$.
Our canonical similarity result states that, in any canonical execution which corresponds to some input configuration $c$,  $\mathcal{A}$ can only decide a value which is admissible for \emph{all} input configurations similar to $c$.
Informally, the reason is that correct processes cannot distinguish silent faulty processes from slow correct ones.

\begin{lemma}[Canonical similarity] \label{lemma:helper_main}
Let $\mathit{val}$ be any solvable validity property and let $\mathcal{A}$ be any algorithm which solves the consensus problem with $\mathit{val}$.
Let $\mathcal{E} \in \mathit{execs}(\mathcal{A})$ be any (potentially infinite) canonical execution and let $\mathsf{input\_conf}(\mathcal{E}) = c$, for some input configuration $c \in \mathcal{I}$.
If a value $v' \in \mathcal{V}_O$ is decided by a correct process in $\mathcal{E}$, then $v' \in \bigcap\limits_{c' \in \mathit{sim}(c)} \mathit{val}(c')$.
\end{lemma}
\begin{proof}
We prove the lemma by contradiction.
Suppose that $v' \notin \bigcap\limits_{c' \in \mathit{sim}(c)} \mathit{val}(c')$.
Hence, there exists an input configuration $c' \in \mathit{sim}(c)$ such that $v' \notin \mathit{val}(c')$.
If $\mathcal{E}$ is infinite, let $\mathcal{E}_P \gets \mathcal{E}$.
Otherwise, let $\mathcal{E}_P$ denote any infinite continuation of $\mathcal{E}$ such that (1) $\mathcal{E}_P$ is canonical, and (2) $\mathsf{input\_conf}(\mathcal{E}_P) = c$.
Let $P$ be any process such that $P \in \processes{c'} \cap \processes{c}$; such a process exists as $c' \sim c$.
As $\mathcal{A}$ satisfies \emph{Termination} and \emph{Agreement}, $\mathcal{E}_P$ is an infinite execution, and $P$ is correct in $\mathcal{E}_P$, $P$ decides $v'$ in $\mathcal{E}_P$.
Let $\mathcal{S}$ denote the set of all processes which take a computational step in $\mathcal{E}_P$ until $P$ decides $v'$; note that $\mathcal{S} \subseteq \processes{c}$ and $P \in \mathcal{S}$.

In the next step, we construct another execution $\mathcal{E}' \in \mathit{execs}(\mathcal{A})$ such that $\mathsf{input\_conf}(\mathcal{E}') = c'$:
\begin{compactenum}
    \item $\mathcal{E}'$ is identical to $\mathcal{E}_P$ until process $P$ decides $v'$.
    
    \item All processes not in $\processes{c'}$ are faulty in $\mathcal{E}'$ (they might behave correctly until $P$ has decided), and all processes in $\processes{c'}$ are correct in $\mathcal{E}'$.
    
    \item After $P$ has decided, processes in $\processes{c'} \setminus{\mathcal{S}}$ ``wake up'' with the proposals specified in $c'$.
    
    \item GST is set to after all processes in $\processes{c'}$ have taken a computational step.
\end{compactenum}
For every process $P_i \in \mathcal{S} \cap \processes{c'}$, the proposal of $P_i$ in $\mathcal{E}'$ is $\mathsf{proposal}(c'[i])$; recall that $c'[i] = c[i]$ as $c' \sim c$ and $\mathcal{S} \subseteq \processes{c}$.
Moreover, for every process $P_j \in \processes{c'} \setminus{\mathcal{S}}$, the proposal of $P_j$ in $\mathcal{E}'$ is $\mathsf{proposal}(c'[j])$ (due to the step 3 of the construction).
Hence, $\mathsf{input\_conf}(\mathcal{E}') = c'$.
Furthermore, process $P$, which is correct in $\mathcal{E}'$, decides $v'$ (due to the step 1 of the construction).
As $v' \notin \mathit{val}(c')$, we reach a contradiction with the fact that $\mathcal{A}$ satisfies $\mathit{val}$, which proves the lemma.
\end{proof}

We underline that \Cref{lemma:helper_main} does not rely on \emph{any} assumptions on $n$ and $t$.
In other words, \Cref{lemma:helper_main} is applicable to any algorithm irrespectively of the parameters $n$ and $t$ (as long as $0 < t < n$).

    
    
    
\section{Necessary solvability conditions}  \label{section:classification}

This section gives the necessary conditions for the solvability of validity properties. 
We start by focusing on the case of $n \leq 3t$: we prove that, if $n \leq 3t$, all solvable validity properties are trivial (\Cref{subsection:triviality_main}).
Then, we consider the case of $n > 3t$: we define the similarity condition, and prove its necessity for solvability (\Cref{subsection:necessary_condition}).

\subsection{Triviality of Solvable Validity Properties if $n \leq 3t$} \label{subsection:triviality_main}

Some validity properties, such as \emph{Weak Validity} and \emph{Strong Validity}, are known to be unsolvable with $n \leq 3t$~\cite{DLS88,PeaseSL80}.
This seems to imply a split of validity properties depending on the resilience threshold.
We prove that such a split indeed exists for $n \leq 3t$, and, importantly, that it applies to \emph{all} solvable validity properties.
Implicitly, this means that there is no ``useful'' relaxation of the validity property that can tolerate $t \geq \lceil n / 3 \rceil$ failures.
Concretely, we prove the following theorem:

\begin{theorem} \label{theorem:triviality_main}
If any validity property $\mathit{val}$ is solvable with $n \leq 3t$, then the validity property is trivial, i.e., there exists a value $v' \in \mathcal{V}_O$ such that $v' \in \bigcap\limits_{c \in \mathcal{I}} \mathit{val}(c)$.
\end{theorem}




Before presenting the proof of the theorem, we introduce the \emph{compatibility relation} between input configurations, which we use throughout this subsection.

\paragraph{Compatibility.}
We define the compatibility relation (``$\diamond$'') between input configurations:
\begin{equation} \nonumber
    \begin{split}
        &\forall c_1, c_2 \in \allconfigurations: c_1 \diamond c_2 \iff \\ 
        \big(|\processes{c_1} \cap \processes{c_2}| \leq t\big) &\land \big(\processes{c_1} \setminus{\processes{c_2}}\neq\emptyset\big) \land \big(\processes{c_2} \setminus{\processes{c_1}}\neq\emptyset\big).
    \end{split}
\end{equation}


That is, $c_1$ is compatible with $c_2$ if and only if (1) there are at most $t$ processes in common, (2) there exists a process which belongs to $c_1$ and does not belong to $c_2$, and (3) there exists a process which belongs to $c_2$ and does not belong to $c_1$.
For example, when $n = 3$ and $t = 1$, $c = \big[ (P_1, 0), (P_2, 0) \big]$ is compatible with $\big[ (P_1, 1), (P_3, 1) \big]$, whereas $c$ is not compatible with $\big[ (P_1, 1), (P_2, 1), (P_3, 1) \big]$.
Observe that the compatibility relation is symmetric and irreflexive.

\paragraph{Proof of \Cref{theorem:triviality_main}.}
Throughout the rest of the subsection, we fix any validity property $\mathit{val}$ which is solvable with $n \leq 3t$; our aim is to prove the triviality of $\mathit{val}$.
Moreover:
\begin{compactitem}
    \item We assume that $n \leq 3t$.

    \item We fix any algorithm $\mathcal{A}$ which solves consensus with $\mathit{val}$.

    \item We fix any input configuration $\mathit{base} \in \mathcal{I}_{n - t}$ with exactly $n - t$ process-proposal pairs.

    \item We fix any infinite canonical execution $\mathcal{E}_{\mathit{base}} \in \mathit{execs}(\mathcal{A})$ such that $\mathsf{input\_conf}(\mathcal{E}_{\mathit{base}}) = \mathit{base}$.
    As $\mathcal{A}$ satisfies \emph{Termination} and $\mathit{val}$, and $\mathcal{E}_{\mathit{base}}$ is infinite, some value $v_{\mathit{base}} \in \mathit{val}(\mathit{base})$ is decided by a correct process in $\mathcal{E}_{\mathit{base}}$.
\end{compactitem}


    
    

First, we show that only $v_{\mathit{base}}$ can be decided in any canonical execution which corresponds to any input configuration compatible with $\mathit{base}$.
If a value different from $v_{\mathit{base}}$ is decided, we can apply the classical partitioning argument~\cite{LSP82}: the adversary causes a disagreement by partitioning processes into two disagreeing groups.

\begin{lemma} \label{lemma:triviality_compatible}
Let $c \in \mathcal{I}$ be any input configuration such that $c \diamond \mathit{base}$.
Let $\mathcal{E}_c \in \mathit{execs}(\mathcal{A})$ be any (potentially infinite) canonical execution such that $\mathsf{input\_conf}(\mathcal{E}_c) = c$.
If a value $v_c \in \mathcal{V}_O$ is decided by a correct process in $\mathcal{E}_c$, then $v_c = v_{\mathit{base}}$.
\end{lemma}
\begin{proof}
By contradiction, suppose that some value $v_c \neq v_{\mathit{base}}$ is decided by a correct process in $\mathcal{E}_c$.
Since $c \diamond \mathit{base}$, there exists a process $Q \in \processes{c}\setminus{\processes{\mathit{base}}}$.
If $\mathcal{E}_c$ is infinite, let $\mathcal{E}_c^Q \gets \mathcal{E}_c$.
Otherwise, let $\mathcal{E}_c^Q$ be an infinite canonical continuation of $\mathcal{E}_c$.
The following holds for process $Q$: (1) $Q$ decides $v_c$ in $\mathcal{E}_c^Q$ (as $\mathcal{A}$ satisfies \emph{Agreement} and \emph{Termination}), and (2) $Q$ is silent in $\mathcal{E}_{\mathit{base}}$.
Let $t_Q$ denote the time at which $Q$ decides in $\mathcal{E}_c^Q$.
Similarly, there exists a process $P \in \processes{\mathit{base}} \setminus{\processes{c}}$; note that (1) $P$ decides $v_{\mathit{base}}$ in $\mathcal{E}_{\mathit{base}}$, and (2) $P$ is silent in $\mathcal{E}_c^Q$.
Let $t_P$ denote the time at which $P$ decides in $\mathcal{E}_{\mathit{base}}$.

In the next step, we construct an execution $\mathcal{E} \in \mathit{execs}(\mathcal{A})$ by ``merging'' $\mathcal{E}_{\mathit{base}}$ and $\mathcal{E}_c^Q$:
\begin{compactenum}
    \item We separate processes into 4 groups: (1) group $A = \pi(\mathit{base}) \setminus{\pi(c)}$, (2) group $B = \pi(\mathit{base}) \cap \pi(c)$, (3) group $C = \pi(c) \setminus{\pi(\mathit{base})}$, and (4) group $D = \Pi \setminus{(A \cup B \cup C)}$.

    \item Processes in $B$ behave towards processes in $A$ as in $\mathcal{E}_{\mathit{base}}$, and towards processes in $C$ as in $\mathcal{E}_c^Q$.

    \item Communication between groups $A$ and $C$ is delayed until after time $\max(t_P, t_Q)$.

    \item Processes in group $D$ ``wake up'' at time $T_D > \max(t_P,t_Q)$ (with any proposals), and behave correctly throughout the entire execution.
    
    \item We set GST to after $T_D$.
\end{compactenum}
The following holds for $\mathcal{E}$:
\begin{compactitem}
    \item Processes in $\Pi \setminus{B}$ are correct in $\mathcal{E}$.
    In other words, only processes in $B$ are faulty in $\mathcal{E}$.
    Recall that $|B| \leq t$ as $\mathit{base} \diamond c$.

    \item Process $Q$, which is correct in $\mathcal{E}$, cannot distinguish $\mathcal{E}$ from $\mathcal{E}_c^Q$ until after time $\max(t_P, t_Q)$.
    Hence, process $Q$ decides $v_c$ in $\mathcal{E}$.

    \item Process $P$, which is correct in $\mathcal{E}$, cannot distinguish $\mathcal{E}$ from $\mathcal{E}_{\mathit{base}}$ until after time $\max(t_P, t_Q)$.
    Hence, process $P$ decides $v_{\mathit{base}} \neq v_c$ in $\mathcal{E}$.
\end{compactitem}
Therefore, we reach a contradiction with the fact that $\mathcal{A}$ satisfies \emph{Agreement}.
Thus, $v_c = v_{\mathit{base}}$.
\end{proof}

Observe that the proposals of an input configuration compatible with $\mathit{base}$ do not influence the decision: given an input configuration $c \in \mathcal{I}$, $c \diamond \mathit{base}$, only $v_{\mathit{base}}$ can be decided in any canonical execution which corresponds to $c$, \emph{irrespectively} of the proposals.

Next, we prove a direct consequence of \Cref{lemma:triviality_compatible}: for every input configuration $c_n \in \mathcal{I}_n$, there exists an infinite execution $\mathcal{E}_{n}$ such that (1) $\mathcal{E}_{n}$ corresponds to $c_n$, and (2) $v_{\mathit{base}}$ is decided in $\mathcal{E}_{n}$.


\begin{lemma} \label{lemma:triviality_complete}
For any input configuration $c_n \in \mathcal{I}_n$, there exists an infinite execution $\mathcal{E}_{n} \in \mathit{execs}(\mathcal{A})$ such that (1) $\mathsf{input\_conf}(\mathcal{E}_{n}) = c_n$, and (2) $v_{\mathit{base}}$ is decided in $\mathcal{E}_{n}$.
\end{lemma}
\begin{proof}
Fix any input configuration $c_n \in \mathcal{I}_n$.
We construct an input configuration $c_{n - t} \in \mathcal{I}_{n - t}$:
\begin{compactenum}
    \item For every process $P_i \notin \processes{\mathit{base}}$, we include a process-proposal pair $(P_i, v)$ in $c_{n - t}$ such that $v = \mathsf{proposal}(c_n[i])$.
    Note that there are $t$ such processes as $|\processes{\mathit{base}}| = n - t$.

    \item We include any $n - 2t$ process-proposal pairs $(P_i, v)$ in $c_{n - t}$ such that (1) $P_i \in \processes{\mathit{base}}$, and (2) $v = \mathsf{proposal}(c_n[i])$.
    That is, we ``complete'' $c_{n - t}$ (constructed in the step 1) with $n - 2t$ process-proposal pairs such that the process is ``borrowed'' from $\mathit{base}$, and its proposal is ``borrowed'' from $c_n$.
\end{compactenum}
Observe that $c_{n - t} \diamond \mathit{base}$ as (1) $|\processes{c_{n - t}} \cap \processes{\mathit{base}}| \leq t$ ($n - 2t \leq t$ when $n \leq 3t$), (2) there exists a process $P \in \processes{\mathit{base}} \setminus \processes{c_{n - t}}$ (while constructing $c_{n - t}$, we excluded $t > 0$ processes from $\mathit{base}$; step 2), and (3) there exists a process $Q \in \processes{c_{n - t}} \setminus{\processes{\mathit{base}}}$ (we included $t > 0$ processes in $\mathit{c_{n - t}}$ which are not in $\mathit{base}$; step 1).

Let $\mathcal{E}_{n - t} \in \mathit{execs}(\mathcal{A})$ denote any infinite canonical execution such that $\mathsf{input\_conf}(\mathcal{E}_{n - t}) = c_{n - t}$.
As $\mathcal{A}$ satisfies \emph{Termination}, some value is decided by correct processes in $\mathcal{E}_{n - t}$; due to \Cref{lemma:triviality_compatible}, that value is $v_{\mathit{base}}$.
Finally, we are able to construct an infinite execution $\mathcal{E}_n \in \mathit{execs}(\mathcal{A})$ such that (1) $\mathsf{input\_conf}(\mathcal{E}_n) = c_n$, and (2) $v_{\mathit{base}}$ is decided in $\mathcal{E}_n$:
\begin{compactenum}
    \item All processes are correct in $\mathcal{E}_n$.

    \item Until some correct process $P \in \processes{c_{n - t}}$ decides $v_{\mathit{base}}$, $\mathcal{E}_n$ is identical to $\mathcal{E}_{n - t}$.
    Let $\mathcal{S}$ denote the set of all processes which take a computational step in $\mathcal{E}_{n - t}$ until $P$ decides $v_{\mathit{base}}$; note that $\mathcal{S} \subseteq \processes{c_{n - t}}$ and $P \in \mathcal{S}$.

    \item Afterwards, every process $Q \notin \mathcal{S}$ ``wakes up'' with the proposal specified in $c_n$.
    
    \item GST occurs after all processes have taken a step.
\end{compactenum}
Therefore, $v_{\mathit{base}}$ is indeed decided in $\mathcal{E}_n$ and $\mathsf{input\_conf}(\mathcal{E}_n) = c_n$, which concludes the proof.
\end{proof}

We are now ready to prove that $\mathit{val}$ is trivial.
We do so by showing that, for any input configuration $c \in \mathcal{I}$, $v_{\mathit{base}} \in \mathit{val}(c)$.

\begin{lemma}
\label{lemma:triviality_main_2} 
Validity property $\mathit{val}$ is trivial.
\end{lemma}
\begin{proof}
We fix any input configuration $c \in \mathcal{I}$.
Let us distinguish two possible scenarios:
\begin{compactitem}
    \item Let $c \in \mathcal{I}_n$.
    There exists an infinite execution $\mathcal{E}_c \in \mathit{execs}(\mathcal{A})$ such that (1) $\mathsf{input\_conf}(\mathcal{E}_c) = c$, and (2) $v_{\mathit{base}}$ is decided in $\mathcal{E}_c$ (by \Cref{lemma:triviality_complete}).
    As $\mathcal{A}$ satisfies $\mathit{val}$, $v_{\mathit{base}} \in \mathit{val}(c)$.

    \item Let $c \notin \mathcal{I}_n$.
    We construct an input configuration $c_n \in \mathcal{I}_n$ in the following way:
    \begin{compactenum}
        \item Let $c_n \gets c$.

        \item For every process $P \notin \processes{c}$, $(P, \text{any proposal})$ is included in $c_n$.
    \end{compactenum}
    Due to the construction of $c_n$, $c_n \sim c$.
    By \Cref{lemma:triviality_complete}, there exists an infinite execution $\mathcal{E}_{n} \in \mathit{execs}(\mathcal{A})$ such that (1) $\mathsf{input\_conf}(\mathcal{E}_n) = c_n$, and (2) $v_{\mathit{base}}$ is decided in $\mathcal{E}_{n}$; $\mathcal{E}_n$ is a canonical execution as all processes are correct.
    Therefore, canonical similarity ensures that $v_{\mathit{base}} \in \mathit{val}(c)$ (\Cref{lemma:helper_main}).
\end{compactitem}
In both possible cases, $v_{\mathit{base}} \in \mathit{val}(c)$.
Thus, the theorem.
\end{proof}

\Cref{lemma:triviality_main_2} concludes the proof of \Cref{theorem:triviality_main}, as \Cref{lemma:triviality_main_2} proves that $\mathit{val}$, any solvable validity property with $n \leq 3t$, is trivial.
\Cref{fig:triviality} depicts the proof of \Cref{theorem:triviality_main}.
Since this subsection shows that consensus cannot be useful when $n \leq 3t$, the rest of the paper focuses on the case of $n > 3t$.

\begin{figure}[ht]
    \centering
    \includegraphics[scale = 0.3]{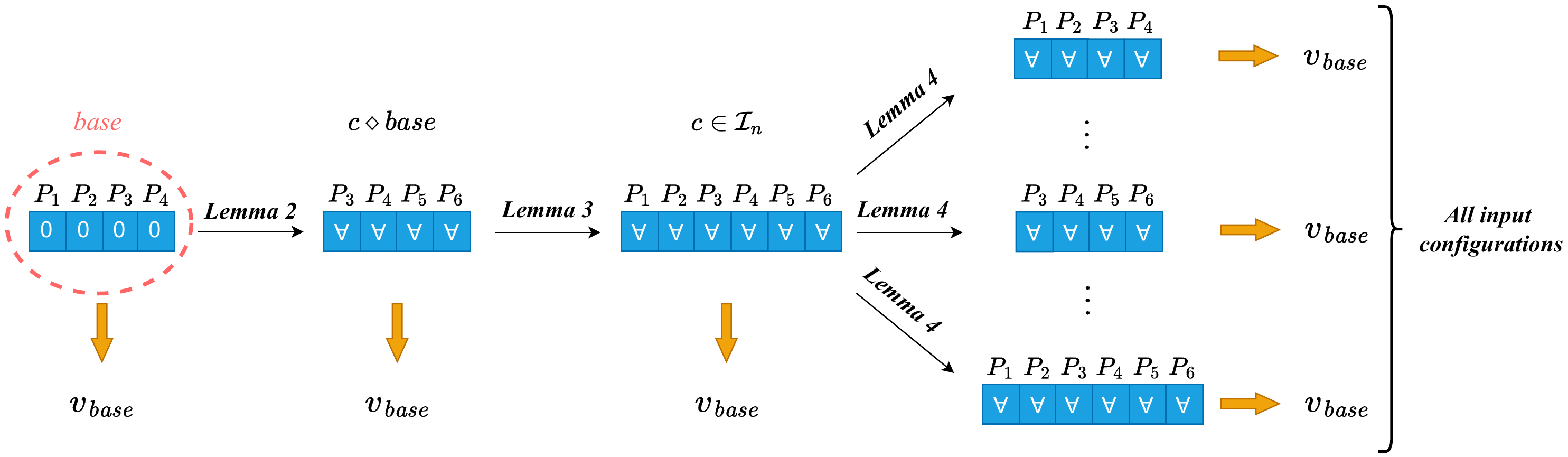}
    \caption{\Cref{theorem:triviality_main}: Overview of the proof in the case of $n = 6$, $t = 2$, and $\mathit{base} = \big[ (P_1, 0), (P_2, 0), (P_3, 0), (P_4, 0) \big]$.}
    \label{fig:triviality}
\end{figure}

\paragraph{Remark.}
As \Cref{theorem:triviality_main} shows, any validity property which is solvable with $n \leq 3t$ is trivial.
However, we now strengthen the aforementioned necessary condition for solvable validity properties with $n \leq 3t$.

\begin{theorem} \label{theorem:triviality_2_main}
If any validity property $\mathit{val}$ is solvable with $n \leq 3t$, then there exists a finite procedure $\mathsf{always\_admissible}$ which returns a value $v' \in \mathcal{V}_O$ such that $v' \in \bigcap\limits_{c \in \mathcal{I}} \mathit{val}(c)$.
\end{theorem}
\begin{proof}
We prove the theorem by contradiction.
Hence, suppose that there exists a validity property $\mathit{val}$ which is solvable with $n \leq 3t$, and that there does not exist a finite procedure $\mathsf{always\_admissible}$ which returns a value $v' \in \mathcal{V}_O$ such that $v' \in \bigcap\limits_{c \in \mathcal{I}} \mathit{val}(c)$.
As $\mathit{val}$ is solvable, there exists an algorithm $\mathcal{A}$ which solves consensus with $\mathit{val}$.

Let us fix any input configuration $\mathit{base} \in \mathcal{I}_{n - t}$ (as done in the proof of \Cref{theorem:triviality_main}).
Moreover, let $v_{\mathit{base}}$ be the value decided in \emph{any} infinite canonical execution $\mathcal{E}_{\mathit{base}} \in \mathit{execs}(\mathcal{A})$ such that $\mathsf{input\_conf}(\mathcal{E}_{\mathit{base}}) = \mathit{base}$; observe that the prefix of $\mathcal{E}_{\mathit{base}}$ is finite as processes can take only finitely many steps in finite time.
As proven in the proof of \Cref{theorem:triviality_main}, $v_{\mathit{base}} \in \mathit{val}(c)$, for any input configuration $c \in \mathcal{I}$.
Hence, there exists a finite procedure $\mathsf{always\_admissible}$ which returns a value admissible according to all input configurations ($\mathcal{A}$ is such a procedure).
Thus, we reach a contradiction, which concludes the proof.
\end{proof}

\Cref{theorem:triviality_2_main} states that, if a validity property is solvable with $n \leq 3t$, not only that the property is trivial (as proven by \Cref{theorem:triviality_main}), but there exists a finite procedure which retrieves an always-admissible value.
Thus, \Cref{theorem:triviality_2_main} strictly extends \Cref{theorem:triviality_main}.
Observe that, if $n \leq 3t$ and a validity property is associated with a finite procedure $\mathsf{always\_admissible}$ which retrieves an always-admissible value, solving consensus with that specific properties is trivial: each process immediately decides the value returned by $\mathsf{always\_admissible}$.
Thus, an existence of the $\mathsf{always\_admissible}$ procedure is a necessary and sufficient condition for solvable validity properties with $n \leq 3t$.

\subsection{Similarity Condition: Necessary Solvability Condition} \label{subsection:necessary_condition}

This subsection defines the similarity condition, and proves its necessity for solvable validity properties.

\begin{definition} [Similarity condition] \label{definition:similarity_condition}
A validity property $\mathit{val}$ satisfies the \emph{similarity condition} ($\mathcal{C}_S$, in short) if and only if there exists a Turing-computable function $\Lambda: \mathcal{I}_{n - t} \to \mathcal{V}_O$ such that:
\begin{equation*}
 \forall c \in \mathcal{I}_{n - t}: \Lambda(c) \in \bigcap\limits_{c' \in \mathit{sim}(c)} \mathit{val}(c').
\end{equation*}
\end{definition}

$\mathcal{C}_S$ states that, for every input configuration $c \in \mathcal{I}_{n - t}$, there exists a Turing-computable function $\Lambda(c)$ which retrieves a common admissible decision among all input configurations similar to $c$.\footnote{A function is Turing-computable if there exists a finite procedure to compute it.}
The necessity of $\mathcal{C}_S$ follows from the canonical similarity result: in any infinite canonical execution, a common admissible value must be decided (\Cref{lemma:helper_main}).

\begin{theorem} \label{theorem:similarity_condition_necessary}
Any solvable validity property satisfies $\mathcal{C}_S$.
\end{theorem}
\begin{proof}
By the means of contradiction, let there exist a validity property $\mathit{val}$ such that (1) $\mathit{val}$ does not satisfy $\mathcal{C}_S$, and (2) $\mathit{val}$ is solvable.
Let $\mathcal{A}$ be any algorithm which solves the Byzantine consensus problem with $\mathit{val}$.
As $\mathit{val}$ does not satisfy $\mathcal{C}_S$, there does not exist a Turing-computable function $\Lambda: \mathcal{I}_{n - t} \to \mathcal{V}_O$ such that, for every input configuration $c \in \mathcal{I}_{n - t}$, $\Lambda(c) \in \bigcap\limits_{c' \in \mathit{sim}(c)} \mathit{val}(c')$.

Fix any input configuration $c \in \mathcal{I}_{n - t}$ for which $\Lambda(c)$ is not defined or not Turing-computable.
Let $\mathcal{E}_c \in \mathit{execs}(\mathcal{A})$ be an infinite canonical execution such that (1) $\mathsf{input\_conf}(\mathcal{E}_c) = c$, (2) the system is synchronous from the very beginning ($\text{GST} = 0$), and (3) the message delays are exactly $\delta$.
In other words, $\mathcal{E}_c$ unfolds in a ``lock-step'' manner.
As $\mathcal{A}$ satisfies \emph{Termination} and $\mathcal{E}_c$ is an infinite execution, some value $v_c \in \mathcal{V}_O$ is decided by a correct process in $\mathcal{E}_c$; observe that the prefix of $\mathcal{E}_c$ in which $v_c$ is decided is finite as processes take only finitely many steps in finite time.
By canonical similarity (\Cref{lemma:helper_main}), $v_c \in \bigcap\limits_{c' \in \similar{c}} \mathit{val}(c')$.
Hence, $\Lambda(c)$ is defined ($\Lambda(c) = v_c$) and Turing-computable ($\mathcal{A}$ computes it).
Therefore, we reach a contradiction with the fact that $\Lambda(c)$ is not defined or not Turing-computable, which concludes the proof.
\end{proof}

Notice that, for proving the necessity of $\mathcal{C}_S$ (\Cref{theorem:similarity_condition_necessary}), we do not rely on the $n > 3t$ assumption.
Thus, $\mathcal{C}_S$ is necessary for \emph{all} solvable validity properties (irrespectively of the resilience threshold).
However, as proven in \Cref{subsection:triviality_main}, $\mathcal{C}_S$ is not sufficient when $n \leq 3t$: e.g., \emph{Weak Validity} satisfies $\mathcal{C}_S$, but it is unsolvable with $n \leq 3t$~\cite{DLS88,PeaseSL80}.
(Observe that any solvable validity property with $n \leq 3t$ satisfies $\mathcal{C}_S$.)

\section{Lower Bound \& General Algorithm}
\label{section:lower_upper_bounds}

First, we show that any non-trivial and solvable validity property requires $\Omega(t^2)$ messages to be exchanged (\Cref{subsection:lower_bound}).
Then, we present \general, a general algorithm that, if $n > 3t$, solves consensus  with any validity property which satisfies $\mathcal{C}_S$ (\Cref{subsection:general_algorithm}).
Thus, \general proves the sufficiency of $\mathcal{C}_S$ when $n > 3t$.

\subsection{Lower Bound on Message Complexity}
\label{subsection:lower_bound}

In this subsection, we prove the following theorem:

\begin{theorem} \label{theorem:lower_bound_main}
If an algorithm solves consensus  with a non-trivial validity property, the message complexity of the algorithm is $\Omega(t^2)$.
\end{theorem}

\Cref{theorem:lower_bound_main} extends the seminal Dolev-Reischuk bound~\cite{dolev1985bounds}, proven only for consensus with \emph{Strong Validity}, to all non-trivial consensus variants.
To prove \Cref{theorem:lower_bound_main}, we intertwine the idea of the original proof~\cite{dolev1985bounds} with canonical similarity (\Cref{lemma:helper_main}).


\paragraph{Proof of \Cref{theorem:lower_bound_main}.}
In our proof, we show that any algorithm which solves Byzantine consensus with a non-trivial validity property has a synchronous execution in which correct processes send 
more than $(\frac{t}{2})^2$ messages.
Hence, throughout the entire subsection, we fix a non-trivial and solvable validity property $\mathit{val}$.
Moreover, we fix $\mathcal{A}$, an algorithm which solves Byzantine consensus with $\mathit{val}$.
As $\mathit{val}$ is a non-trivial validity property, $n > 3t$ (\Cref{subsection:triviality_main}).

Next, we define a specific infinite execution $\mathcal{E}_{\mathit{base}} \in \mathit{execs}(\mathcal{A})$:
\begin{compactenum}
     \item $\text{GST} = 0$.
     That is, the system is synchronous throughout the entire execution.
     
     \item All processes are separated into two disjoint groups: (1) group $A$, with $|A| = n - \lceil \frac{t}{2} \rceil$, and (2) group $B$, with $|B| = \lceil \frac{t}{2} \rceil$.
     
     \item All processes in the group $A$ are correct, whereas all processes in the group $B$ are faulty.
     
     \item We fix any value $v^* \in \mathcal{V}_I$.
     For every correct process $P_A \in A$, the proposal of $P_A$ in $\mathcal{E}_{\mathit{base}}$ is $v^*$.
     
     \item For every faulty process $P_B \in B$, $P_B$ behaves correctly in $\mathcal{E}_{\mathit{base}}$ with its proposal being $v^*$, except that (1) $P_B$ ignores the first $\lceil \frac{t}{2} \rceil$ messages received from other processes, and (2) $P_B$ omits sending messages to other processes in $B$. 
\end{compactenum}
To prove \Cref{theorem:lower_bound_main}, it suffices to show that the message complexity of $\mathcal{E}_{\mathit{base}}$ is greater than $ (\lceil \frac{t}{2} \rceil)^2$.
By contradiction, let the correct processes (processes in $A$) send $\leq (\lceil \frac{t}{2} \rceil)^2$ messages in $\mathcal{E}_{\mathit{base}}$.

The first step of our proof shows that, given that correct processes send $\leq (\lceil \frac{t}{2} \rceil)^2$ messages in $\mathcal{E}_{\mathit{base}}$, there must exist a process $Q \in B$ which can correctly decide some value $v_Q \in \mathcal{V}_O$ without receiving \emph{any} message from any other process.
We prove this claim using the pigeonhole principle.

\begin{lemma} \label{lemma:behavior_without_receptions}
There exist a value $v_Q \in \mathcal{V}_O$ and a process $Q \in B$ such that $Q$ has a correct 
local behavior $\beta_Q$ in which (1) $Q$ decides $v_Q$, and (2) $Q$ receives no messages from other processes.
\end{lemma}
\begin{proof}
By assumption, correct processes (i.e., processes in the group $A$) send $\leq (\lceil \frac{t}{2} \rceil)^2$ messages in $\mathcal{E}_{\mathit{base}}$.
Therefore, due to the pigeonhole principle, there exists a process $Q \in B$ which receives at most $\lceil \frac{t}{2} \rceil$ messages (from other processes) in $\mathcal{E}_{\mathit{base}}$.
Recall that $Q$ behaves correctly in $\mathcal{E}_{\mathit{base}}$ with its proposal being $v^* \in \mathcal{V}_I$, except that (1) $Q$ ignores the first $\lceil \frac{t}{2} \rceil$ messages received from other processes, and (2) $Q$ does not send any messages to other processes in the group $B$.
We denote by $S_Q$ the set of processes, not including $Q$, which send messages to $Q$ in $\mathcal{E}_{\mathit{base}}$; $|S_Q| \leq \lceil \frac{t}{2} \rceil$.

Next, we construct an infinite execution $\mathcal{E}_{\mathit{base}}'$.
Execution $\mathcal{E}_{\mathit{base}}'$ is identical to $\mathcal{E}_{\mathit{base}}$, except that:
\begin{compactenum}
    
    \item Processes in $(A \cup \{Q\}) \setminus{S_Q}$ are correct; other processes are faulty.
    That is, we make $Q$ correct in $\mathcal{E}'_{\mathit{base}}$, and we make all processes in $S_Q$ faulty in $\mathcal{E}'_{\mathit{base}}$.
    
    \item Processes in $B \setminus{\{Q\}}$ behave exactly as in $\mathcal{E}_{\mathit{base}}$.
    Moreover, processes in $S_Q$ behave exactly as in $\mathcal{E}_{\mathit{base}}$, except that they \emph{do not} send any message to $Q$.
\end{compactenum}
Due to the construction of $\mathcal{E}_{\mathit{base}}'$, process $Q$ does not receive any message (from any other process) in $\mathcal{E}_{\mathit{base}}'$.
As $Q$ is correct in $\mathcal{E}_{\mathit{base}}'$ and $\mathcal{A}$ satisfies \emph{Termination}, $Q$ decides some value $v_Q \in \mathcal{V}_O$ in $\mathcal{E}_{\mathit{base}}'$.
Thus, $Q$ has a correct local behavior $\beta_Q$ in which it decides $v_Q \in \mathcal{V}_O$ without having received messages from other processes. 
\end{proof}


The previous proof concerns deterministic protocols and uses a deterministic adversarial strategy.
We invite the reader to~\Cref{section:nondeterminism_appendix} for a remark on non-determinism.

In the second step of our proof, we show that there exists an infinite execution in which (1) $Q$ is faulty and silent, and (2) correct processes decide some value $v \neq v_Q$.

\begin{lemma} \label{lemma:conflicting_execution}
There exists an infinite execution $\mathcal{E}_{v}$ such that (1) $Q$ is faulty and silent in $\mathcal{E}_{v}$, and (2) a value $v \neq v_Q$ is decided by a correct process.
\end{lemma}
\begin{proof}
As $\mathit{val}$ is a non-trivial validity property, there exists an input configuration $c_{\notni v_Q} \in \mathcal{I}$ such that $v_Q \notin \mathit{val}(c_{\notni v_Q})$; recall that $v_Q$ is the value that $Q$ can correctly decide without having received any message from any other process (\Cref{lemma:behavior_without_receptions}).
We consider two possible cases:
\begin{compactitem}
     \item Let $Q \notin \processes{c_{\notni v_Q}}$.
     Thus, $\mathcal{E}_v$ is any infinite canonical execution which corresponds to $c_{\notni v_Q}$.
     As $v_Q \notin \mathit{val}(c_{\notni v_Q})$, the value decided in $\mathcal{E}_v$ must be different from $v_Q$ (as $\mathcal{A}$ satisfies $\mathit{val}$).
     
     \item Let $Q \in \processes{c_{\notni v_Q}}$.
     We construct an input configuration $c_{\notni Q} \in \mathcal{I}$ such that $Q \notin \processes{c_{\notni Q}}$:
     \begin{compactenum}
        \item Let $c_{\notni Q} \gets c_{\notni v_Q}$.

        \item We remove $(Q, \cdot)$ from $c_{\notni Q}$.
        That is, we remove $Q$'s process-proposal pair from $c_{\notni Q}$.

        \item If $|\processes{c_{\notni v_Q}}| = n - t$, we add $(Z, \text{any proposal})$ to $c_{\notni Q}$, where $Z$ is any process such that $Z \notin \processes{c_{\notni v_Q}}$; note that such a process $Z$ exists as $t > 0$.
    \end{compactenum}
    Due to the construction of $c_{\notni Q}$, $c_{\notni Q} \sim c_{\notni v_Q}$.
    Indeed, (1) $\processes{c_{\notni Q}} \cap \processes{c_{\notni v_Q}} \neq \emptyset$ (as $n - t - 1 > 0$ when $n > 3t$ and $t > 0$), and (2) for every process $P \in \processes{c_{\notni Q}} \cap \processes{c_{\notni v_Q}}$, the proposal of $P$ is identical in $c_{\notni Q}$ and $c_{\notni v_Q}$.
    
    In this case, $\mathcal{E}_v$ is any infinite canonical execution such that $\mathsf{input\_conf}(\mathcal{E}_{v}) = c_{\notni Q}$.
    As $\mathcal{A}$ satisfies \emph{Termination} and $\mathcal{E}_{v}$ is infinite, some value $v \in \mathcal{V}_O$ is decided by a correct process in $\mathcal{E}_{v}$.
    As $c_{\notni Q} \sim c_{\notni v_Q}$, $v \in \mathit{val}(c_{\notni v_Q})$ (by canonical similarity; \Cref{lemma:helper_main}).
    Finally, $v \neq v_Q$ as (1) $v \in \mathit{val}(c_{\notni v_Q})$, and (2) $v_Q \notin \mathit{val}(c_{\notni v_Q})$.
\end{compactitem}
The lemma holds as its statement is true in both possible cases.
\end{proof}

As we have shown the existence of $\mathcal{E}_{v}$ (\Cref{lemma:conflicting_execution}), we can ``merge'' $\mathcal{E}_{v}$ with the valid local behavior $\beta_Q$ in which $Q$ decides $v_Q$ without having received any message (\Cref{lemma:behavior_without_receptions}).
Hence, we can construct an execution in which $\mathcal{A}$ violates \emph{Agreement}.
Thus, correct processes must send more than $(\lceil \frac{t}{2} \rceil)^2 \in \Omega(t^2)$ messages in $\mathcal{E}_{\mathit{base}}$.

\begin{lemma} \label{lemma:final_lemma_lower_bound}
The message complexity of $\mathcal{E}_{\mathit{base}}$ is greater than $(\lceil \frac{t}{2} \rceil)^2$.
\end{lemma}
\begin{proof}
By \Cref{lemma:behavior_without_receptions}, there exists a local behavior $\beta_Q$ of process $Q$ in which $Q$ decides a value $v_Q$ without having received any message (from any other process).
Let $t_Q$ denote the time at which $Q$ decides $v_Q$ in $\beta_Q$.
Moreover, there exists an infinite execution $\mathcal{E}_{v}$ in which (1) $Q$ is faulty and silent, and (2) correct processes decide a value $v \neq v_Q$ (by \Cref{lemma:conflicting_execution}).
Let $t_{v}$ denote the time at which a correct process decides $v \neq v_Q$ in $\mathcal{E}_{v}$.

We now construct an execution $\mathcal{E}$ in the following way:
\begin{compactenum}
    \item Processes in $\mathit{Corr}_{\mathcal{A}}(\mathcal{E}_{v}) \cup \{Q\}$ are correct in $\mathcal{E}$.
    All other processes are faulty.
    
    \item All messages from and to $Q$ are delayed until after $\max(t_Q, t_{v})$.
    
    \item Process $Q$ exhibits the local behavior $\beta_Q$.
    
    \item Until $\max(t_Q, t_{v})$, no process in $\mathit{Corr}_{\mathcal{A}}(\mathcal{E}_{v})$ can distinguish $\mathcal{E}$ from $\mathcal{E}_{v}$.
    
    \item GST is set to after $\max(t_Q, t_{v})$ (and after all correct processes from the $\mathit{Corr}_{\mathcal{A}}(\mathcal{E}_v) \cup \{Q\}$ set have taken a step).
\end{compactenum}
As no process in $\mathit{Corr}_{\mathcal{A}}(\mathcal{E}_{v})$ can distinguish $\mathcal{E}$ from $\mathcal{E}_{v}$ until time $\max(t_Q, t_{v})$, $v \neq v_Q$ is decided by a correct process in $\mathcal{E}$.
Moreover, $Q$ decides $v_Q$ in $\mathcal{E}$ as it exhibits $\beta_Q$ (step 3 of the construction).
Thus, \emph{Agreement} is violated in $\mathcal{E}$, which contradicts the fact that $\mathcal{A}$ satisfies \emph{Agreement}.
Hence, the starting assumption is not correct: in $\mathcal{E}_{\mathit{base}}$, correct processes send more than $(\lceil \frac{t}{2} \rceil)^2$ messages.
\end{proof}

The following subsection shows that the quadratic bound on message complexity is tight with $t \in \Omega(n)$: \general exchanges $O(n^2)$ messages when relying on a PKI.
We underline that our lower bound holds even for algorithms which employ digital signatures.
Achieving the optimal quadratic message complexity without relying on digital signatures remains an important open question.





\subsection{General Algorithm \general: Similarity Condition is Sufficient if $n > 3t$} \label{subsection:general_algorithm}

In this subsection, we show that $\mathcal{C}_S$ is sufficient for a validity property to be solvable when $n > 3t$.
Furthermore, we prove that, assuming a PKI, the quadratic lower bound (\Cref{subsection:lower_bound}) is tight with $t \in \Omega(n)$.
In brief, we prove the following theorem:

\begin{theorem} \label{theorem:general}
Let $n > 3t$, and let $\mathit{val}$ be any validity property which satisfies $\mathcal{C}_S$.
Then, $\mathit{val}$ is solvable.
Moreover, assuming a public-key infrastructure, there exists an algorithm which solves Byzantine consensus with $\mathit{val}$, and has $O(n^2)$ message complexity.
\end{theorem}

To prove \Cref{theorem:general}, we present \general, an algorithm which solves the Byzantine consensus problem with \emph{any} validity property satisfying $\mathcal{C}_S$, given that $n > 3t$.
In other words, \general solves consensus with any solvable and non-trivial validity property.
Notably, assuming a PKI, \general achieves $O(n^2)$ message complexity, making it optimal (when $t \in \Omega(n)$) according to our quadratic lower bound.

To construct \general, we rely on vector consensus~\cite{neves2005solving,vaidya2013byzantine,Doudou1998,correia2006consensus} (see \Cref{subsubsection:interactive_consistency}), a problem which requires correct processes to agree on the proposals of $n - t$ processes.
Specifically, when a correct process decides a vector $\mathit{vec}$ of $n - t$ proposals (from vector consensus), it decides from \general the common admissible value for all input configurations similar to $\mathit{vec}$, i.e., the process decides $\Lambda(\mathit{vec})$.
Note that the idea of solving consensus from vector consensus is not novel~\cite{Ben-Or94,CGL18,mostefaoui2000binary}.
For some validity properties it is even natural, such as \emph{Strong Validity}  (choose the most common value) or \emph{Weak Validity} (choose any value).
However, thanks to the necessity of $\mathcal{C}_S$ (\Cref{subsection:necessary_condition}), \emph{any} solvable  consensus variant can reuse this simple algorithmic design.

In this subsection, we first recall vector consensus (\Cref{subsubsection:interactive_consistency}).
Then, we utilize vector consensus to construct \general (\Cref{subsubsection:Complete}).
Throughout the entire subsection, $n > 3t$.

\subsubsection{Vector Consensus.} \label{subsubsection:interactive_consistency}
In essence, vector consensus allows each correct process to infer the proposals of $n - t$ (correct or faulty) processes.
Formally, correct processes agree on input configurations (of vector consensus) with exactly $n - t$ process-proposal pairs: $\mathcal{V}_O = \mathcal{I}_{n - t}$. 
Let us precisely define \emph{Vector Validity}, the validity property of vector consensus:
\begin{compactitem}
     \item \emph{Vector Validity:} Let a correct process decide $\mathit{vector} \in \mathcal{V}_O$, which contains exactly $n - t$ process-proposal pairs, such that (1) $(P, v)$ belongs to $\mathit{vector}$, for some process $P \in \Pi$ and some value $v \in \mathcal{V}_I$, and (2) $P$ is a correct process.
     Then, $P$ proposed $v$ to vector consensus.
\end{compactitem}
Intuitively, \emph{Vector Validity} states that, if a correct process ``concludes'' that a value $v$ was proposed by a process $P$ and $P$ is correct, then $P$'s proposal was indeed $v$.

We provide two implementations of vector consensus: (1) a non-authenticated implementation (without any cryptographic primitives), and (2) an authenticated implementation (with digital signatures).
We give the pseudocode of the non-authenticated variant in \Cref{subsection:nonauthenticated_vector}.
The pseudocode of the authenticated variant is presented in \Cref{algorithm:interactive_consistency}.
This variant relies on \textsc{Quad}, a Byzantine consensus algorithm recently introduced in~\cite{civit2022byzantine}; we briefly discuss \textsc{Quad} below.


\paragraph{\textsc{Quad}.}
In essence, \textsc{Quad} is a partially-synchronous, ``leader-based'' Byzantine consensus algorithm, which achieves $O(n^2)$ message complexity.
Internally, \textsc{Quad} relies on a PKI.\footnote{In fact, \textsc{Quad} relies on a threshold signature scheme~\cite{Shoup00}, and not on a PKI. However, by inserting digital signatures in place of threshold signatures, \textsc{Quad} is modified to accommodate for a PKI only while preserving its quadratic message complexity.}
Formally, \textsc{Quad} is concerned with two sets: (1) $\mathcal{V}_{\textsc{Quad}}$, a set of values, and (2) $\mathcal{P}_{\textsc{Quad}}$, a set of proofs.
In \textsc{Quad}, processes propose and decide value-proof pairs.
There exists a function $\mathsf{verify}: \mathcal{V}_{\textsc{Quad}} \times \mathcal{P}_{\textsc{Quad}} \to \{ \mathit{true}, \mathit{false} \}$.
Importantly, $\mathcal{P}_{\textsc{Quad}}$ is not known a-priori: it is only assumed that, if a correct process proposes a pair $(v \in \mathcal{V}_{\textsc{Quad}}, \Sigma \in \mathcal{P}_{\textsc{Quad}})$, then $\mathsf{verify}(v, \Sigma) = \mathit{true}$.
\textsc{Quad} guarantees the following: if a correct process decides a pair $(v, \Sigma)$, then $\mathsf{verify}(v, \Sigma) = \mathit{true}$.
In other words, correct processes decide only valid value-proof pairs.
(See~\cite{civit2022byzantine} for the full details on \textsc{Quad}.)


In our authenticated implementation of vector consensus
(\Cref{algorithm:interactive_consistency}), we rely on a specific instance of \textsc{Quad} where (1) $\mathcal{V}_{\textsc{Quad}} = \mathcal{I}_{n - t}$ (processes propose to \textsc{Quad} the input configurations of vector consensus), and (2) $\mathcal{P}_{\textsc{Quad}}$ is a set of $n - t$ \textsc{proposal} messages (sent by processes in vector consensus).
Finally, given an input configuration $\mathit{vector} \in \mathcal{V}_{\textsc{Quad}}$ and a set of messages $\Sigma \in \mathcal{P}_{\textsc{Quad}}$, $\mathsf{verify}(\mathit{vector}, \Sigma) = \mathit{true}$ if and only if, for every process-proposal pair $(P_j, v_j)$ which belongs to $\mathit{vector}$, $\langle \textsc{proposal}, v_j \rangle_{\sigma_j} \in \Sigma$ (i.e., every process-proposal pair of $\mathit{vector}$ is accompanied by a properly signed \textsc{proposal} message).

\paragraph{Description of authenticated vector consensus (\Cref{algorithm:interactive_consistency}).}
When a correct process $P_i$ proposes a value $v \in \mathcal{V}_I$ to vector consensus (line~\ref{line:propose}), the process broadcasts a signed \textsc{proposal} message (line~\ref{line:broadcast_proposal}).
Once $P_i$ receives $n - t$ \textsc{proposal} messages (line~\ref{line:received_n_t_proposals}), $P_i$ constructs an input configuration $\mathit{vector}$ (line~\ref{line:construct_input_configuration}), and a proof $\Sigma$ (line~\ref{line:construct_proof}) from the received \textsc{proposal} messages.
Moreover, $P_i$ proposes $(\mathit{vector}, \Sigma)$ to \textsc{Quad} (line~\ref{line:propose_quad}).
Finally, when $P_i$ decides a pair $(\mathit{vector}', \Sigma')$ from \textsc{Quad} (line~\ref{line:decide_quad}), $P_i$ decides $\mathit{vector}'$ from vector consensus (line~\ref{line:decide}).

\begin{algorithm}
\caption{Authenticated Vector Consensus: Pseudocode (for process $P_i$)}
\label{algorithm:interactive_consistency}
\footnotesize
\begin{algorithmic} [1]
\State \textbf{Uses:}
\State \hskip2em Best-Effort Broadcast~\cite{cachin2011introduction}, \textbf{instance} $\mathit{beb}$ \BlueComment{no guarantees with faulty sender}
\State \hskip2em \textsc{Quad}~\cite{civit2022byzantine}, \textbf{instance} $\mathit{quad}$

\medskip
\State \textbf{upon} $\mathsf{init}$:
\State \hskip2em $\mathsf{Integer}$ $\mathit{received\_proposals}_i \gets 0$ \BlueComment{the number of received proposals}
\State \hskip2em $\mathsf{Map}(\mathsf{Process} \to \mathcal{V}_I)$ $\mathit{proposals}_i \gets \text{empty}$ \BlueComment{received proposals}
\State \hskip2em $\mathsf{Map}(\mathsf{Process} \to \mathsf{Message})$ $\mathit{messages}_i \gets \text{empty}$ \BlueComment{received messages}

\medskip
\State \textbf{upon} $\mathsf{propose}(v \in \mathcal{V}_I)$: \label{line:propose}
\State \hskip2em \textbf{invoke} $\mathit{beb}.\mathsf{broadcast}\big( \langle \textsc{proposal}, v \rangle_{\sigma_i} \big)$ \BlueComment{broadcast a signed proposal} \label{line:broadcast_proposal}

\medskip
\State \textbf{upon} reception of $\mathsf{Message}$ $m = \langle \textsc{proposal}, v_j \in \mathcal{V}_I \rangle_{\sigma_j}$ from process $P_j$ and $\mathit{received\_proposals}_i < n - t$:
\State \hskip2em $\mathit{received\_proposals}_i \gets \mathit{received\_proposals}_i + 1$
\State \hskip2em $\mathit{proposals}_i[P_j] \gets v_j$
\State \hskip2em $\mathit{messages}_i[P_j] \gets m$

\State \hskip2em \textbf{if} $\mathit{received\_proposals}_i = n - t$: \label{line:received_n_t_proposals} \BlueComment{able to propose to \textsc{Quad}}
\State \hskip4em $\mathsf{Input\_Configuration}$ $\mathit{vector} \gets $ constructed from $\mathit{proposals}_i$ \label{line:construct_input_configuration}
\State \hskip4em $\mathsf{Proof}$ $\Sigma \gets $ all \textsc{proposal} messages from $\mathit{messages}_i$ \label{line:construct_proof}
\State \hskip4em \textbf{invoke} $\mathit{quad}.\mathsf{propose}\big( (\mathit{vector}, \Sigma) \big)$ \label{line:propose_quad}

\medskip
\State \textbf{upon} $\mathit{quad}.\mathsf{decide}\big( (\mathsf{Input\_Configuration} \text{ } \mathit{vector}', \mathsf{Proof} \text{ } \Sigma') \big)$: \label{line:decide_quad}
\State \hskip2em \textbf{trigger} $\mathsf{decide}(\mathit{vector}')$ \label{line:decide}

\end{algorithmic}
\end{algorithm}

The message complexity of \Cref{algorithm:interactive_consistency} is $O(n^2)$ as (1) processes only broadcast \textsc{proposal} messages, and (2) the message complexity of \textsc{Quad} is $O(n^2)$.
We delegate the full proof of the correctness and complexity of \Cref{algorithm:interactive_consistency} to \Cref{subsection:authenticated_vector_proof}.

\subsubsection{\general.} \label{subsubsection:Complete}
We construct \general (\Cref{algorithm:general}) directly from vector consensus.
When a correct process $P_i$ proposes to \general (line~\ref{line:propose_general}), the proposal is forwarded to vector consensus (line~\ref{line:propose_ic}).
Once $P_i$ decides an input configuration $\mathit{vector}$ from vector consensus (line~\ref{line:decide_ic}), $P_i$ decides $\Lambda(\mathit{vector})$ (line~\ref{line:decide_general}).

Note that our implementation of \general (\Cref{algorithm:general}) is independent of the actual implementation of vector consensus.
Thus, by employing our authenticated implementation of vector consensus (\Cref{algorithm:interactive_consistency}), we obtain a general consensus algorithm with $O(n^2)$ message complexity.
On the other hand, by employing a non-authenticated implementation of vector consensus (see \Cref{subsection:nonauthenticated_vector}), we obtain a non-authenticated version of \general, which implies that any validity property which satisfies $\mathcal{C}_S$ is solvable even in a non-authenticated setting (if $n > 3t$).

\begin{algorithm} [ht]
\caption{\general: Pseudocode (for process $P_i$)}
\label{algorithm:general}
\footnotesize
\begin{algorithmic} [1]
\State \textbf{Uses:}
\State \hskip2em Vector Consensus, \textbf{instance} $\mathit{vec\_cons}$

\medskip
\State \textbf{upon} $\mathsf{propose}(v \in \mathcal{V}_I)$: \label{line:propose_general}
\State \hskip2em \textbf{invoke} $\mathit{vec\_cons}.\mathsf{propose}(v)$ \label{line:propose_ic}

\medskip
\State \textbf{upon} $\mathit{vec\_cons}.\mathsf{decide}(\mathsf{Input\_Configuration} \text{ } \mathit{vector})$: \label{line:decide_ic}
\State \hskip2em \textbf{trigger} $\mathsf{decide}\big( \Lambda(\mathit{vector}) \big)$ \label{line:decide_general}
\end{algorithmic}
\end{algorithm}


Finally, we show that \general (\Cref{algorithm:general}) is a general Byzantine consensus algorithm and that its authenticated variant achieves $O(n^2)$ message complexity, which proves \Cref{theorem:general}.

\begin{lemma} \label{theorem:general_algorithm}
Let $\mathit{val}$ be any validity property which satisfies $\mathcal{C}_S$, and let $n > 3t$.
\general solves Byzantine consensus with $\mathit{val}$.
Moreover, if \general employs \Cref{algorithm:interactive_consistency} as its vector consensus building block, the message complexity of \general is $O(n^2)$.
\end{lemma}
\begin{proof}
\emph{Termination} and \emph{Agreement} of \general follow from \emph{Termination} and \emph{Agreement} of vector consensus, respectively.
Moreover, the message complexity of \general is identical to the message complexity of its vector consensus building block.

Finally, we prove that \general satisfies $\mathit{val}$.
Consider any execution $\mathcal{E}$ of \general; let $\mathsf{input\_conf}(\mathcal{E}) = c^*$, for some input configuration $c^* \in \mathcal{I}$.
Moreover, let $\mathit{vector} \in \mathcal{I}_{n - t}$ be the input configuration correct processes decide from vector consensus in $\mathcal{E}$ (line~\ref{line:decide_ic}).
As vector consensus satisfies \emph{Vector Validity}, we have that, for every process $P \in \processes{c^*} \cap \processes{\mathit{vector}}$, $P$'s proposals in $c^*$ and $\mathit{vector}$ are identical.
Hence, $\mathit{vector} \sim c^*$.
Therefore, $\Lambda(\mathit{vector}) \in \mathit{val}(c^*)$ (by the definition of the $\Lambda$ function).
Thus, $\mathit{val}$ is satisfied by \general.
\end{proof}

As \general (\Cref{algorithm:general}) solves the Byzantine consensus problem with any validity property which satisfies $\mathcal{C}_S$ (\Cref{theorem:general_algorithm}) if $n > 3t$, $\mathcal{C}_S$ is sufficient for solvable validity properties when $n > 3t$.
Lastly, as \general relies on vector consensus, we conclude that \emph{Vector Validity} is a \emph{strongest} validity property.
That is, a consensus solution to any solvable variant of the validity property can be obtained (with no additional cost) from vector consensus.

\paragraph{A note on the communication complexity of vector consensus.}
While the version of \general which employs \Cref{algorithm:interactive_consistency} (as its vector consensus building block) has optimal message complexity, its communication complexity is $O(n^3)$ as \textsc{Quad}'s communication complexity is $O(n^3)$ when proofs are of linear size.\footnote{The communication complexity denotes the number of sent words, where a word contains a constant number of values and signatures.}
This presents a linear gap to the lower bound for communication complexity (also $\Omega(n^2)$, implied by \Cref{theorem:lower_bound_main}), and to known optimal solutions for some validity properties (e.g., \emph{Strong Validity}, proven to be $\Theta(n^2)$ \cite{civit2022byzantine,lewis2022quadratic}).
At first glance, this seems like an issue inherent to vector consensus: the decided vectors are linear in size, suggesting that the linear gap could be inevitable.
However, this is not the case.
In \Cref{subsection:better_communication}, we give a vector consensus algorithm with $O(n^2\log n)$ communication complexity, albeit with exponential latency.\footnote{Both our authenticated (\Cref{algorithm:interactive_consistency}) and our non-authenticated (see \Cref{subsection:nonauthenticated_vector}) variants of vector consensus have linear latency, which implies linear latency of \general when employing any of these two algorithms.}
Is it possible to construct vector consensus with subcubic communication and polynomial latency?
This is an important open question, as positive answers would lead to (practical) performance improvements of all consensus variants.
\section{Concluding Remarks} \label{section:conclusion}

This paper studies the validity property of partially synchronous Byzantine consensus.
Namely, we mathematically formalize validity properties, and give the necessary and sufficient conditions for a validity property to be solvable (i.e., for the existence of an algorithm which solves a consensus problem defined with that validity property, in addition to \emph{Agreement} and \emph{Termination}).
Moreover, we prove a quadratic lower bound on the message complexity for all non-trivial (and solvable) validity properties.
Previously, this bound was mainly known for \emph{Strong Validity}.
Lastly, we introduce \general, a general algorithm for  consensus with any solvable (and non-trivial) validity property; assuming a PKI, \general achieves $O(n^2)$ message complexity, showing that the aforementioned lower bound is tight (with $t \in \Omega(n)$).

A natural extension of this work is its adaptation to synchronous environments.
Similarly, can we extend our results to randomized protocols?
Furthermore, investigating consensus variants in which ``exact'' agreement among correct processes is not required (such as approximate~\cite{AbrahamAD04,MendesH13,GhineaLW22,ghinea2023multidimensional} or $k$-set~\cite{BouzidIR16,Delporte-Gallet20,Delporte-Gallet22,lynch1996distributed} agreement) constitutes another important research direction for the future.

Finally, we restate the question posed at the end of \Cref{subsection:general_algorithm}.
Is it possible to solve vector consensus with $o(n^3)$ exchanged bits and polynomial latency?
Recall that, due to the design of \general (\Cref{subsection:general_algorithm}), any (non-trivial) consensus variant can be solved using vector consensus without additional cost.
Therefore, an upper bound on the complexity of vector consensus is an upper bound on the complexity of any consensus variant.
Hence, lowering the $O(n^3)$ communication complexity of vector consensus (while preserving polynomial latency) constitutes an important future research direction.

\begin{acks}
We thank the anonymous reviewers for their insightful comments.
We also thank our colleagues Nirupam Gupta, Matteo Monti, Rafael Pinot and Pierre-Louis Roman for the helpful discussions and comments.
This work was funded in part by the Hasler Foundation (\#21084), the Singapore grant MOE-T2EP20122-0014, and the ARC Future Fellowship program (\#180100496).
\end{acks}

\bibliographystyle{ACM-Reference-Format}
\bibliography{references}

\appendix
\section{Remark on non-determinism}
\label{section:nondeterminism_appendix}

The reader might be aware of several randomized protocols that achieve expected sub-quadratic communication with a negligible probability of failure~\cite{DBLP:conf/podc/AbrahamCDNP0S19, Cohen2020, Boyle2021, Chopard2021, Bhangale2022}.
This might seem at odds with our lower bound (\Cref{theorem:lower_bound_main}).
However, this is not the case.

When proving \Cref{lemma:behavior_without_receptions}, we show that the adversary has a (deterministic) strategy given any deterministic protocol.
In particular, the deterministic strategy of the adversary described in the construction of $\mathcal{E}_{\mathit{base}}$ (including GST=0) resolves the \emph{pure} non-determinism (e.g., network scheduling) and thus deterministically implies $\mathcal{E}_{\mathit{base}}$ and the associated set $S_Q$ (similarly for $\mathcal{E}'_{\mathit{base}}$).
However, we did not demonstrate that an adversary always has a pair of strategies for every \emph{randomized} protocol such that the lemma would hold with a non-negligible probability.
Hence, there is no contradiction whatsoever with the aforementioned randomized protocols.

Extending \Cref{theorem:lower_bound_main} to the randomized case with an adaptive adversary (in the same vein as~\cite{DBLP:conf/podc/AbrahamCDNP0S19}) is left for future work.
\section{Vector Consensus: Proofs \& Omitted Algorithms} \label{section:algorithm_appendix}

In \Cref{subsection:authenticated_vector_proof}, we prove the correctness and complexity of our authenticated implementation of vector consensus (\Cref{algorithm:interactive_consistency}).
We dedicate \Cref{subsection:nonauthenticated_vector} to a non-authenticated implementation of vector consensus.
Finally, in \Cref{subsection:better_communication}, we give an implementation of vector consensus with $O(n^2 \log n)$ communication complexity.
Throughout the entire section, $n > 3t$.


\subsection{Authenticated Implementation (\Cref{algorithm:interactive_consistency}): Proofs} \label{subsection:authenticated_vector_proof}

We start this subsection with some clarifications about \textsc{Quad}~\cite{civit2022byzantine}, a Byzantine consensus algorithm utilized in \Cref{algorithm:interactive_consistency}.
Then, we prove the correctness and complexity of \Cref{algorithm:interactive_consistency}.

\paragraph{A note on \textsc{Quad}.}
In \Cref{subsubsection:interactive_consistency}, we claim that \textsc{Quad} satisfies the following validity property: if a correct process decides a value-proof pair $(v, \Sigma)$, then $\mathsf{verify}(v, \Sigma) = \mathit{true}$.
Technically speaking, the authors of \textsc{Quad} only consider \emph{Weak Validity}, i.e., they do not claim that the their protocol satisfies the aforementioned validity property.
However, modifying their protocol to accommodate for the aforementioned property is trivial: each correct process simply discards each message which contains a pair $(v, \Sigma)$ for which $\mathsf{verify}(v, \Sigma) = \mathit{false}$.

Another subtle remark is that the authors of \textsc{Quad} prove its message and latency complexity assuming that all correct processes start executing \textsc{Quad} by GST.
In \Cref{algorithm:interactive_consistency}, this might not be the case: correct processes might receive $n - t$ \textsc{proposal} messages at $\text{GST} + \delta$ (line~\ref{line:received_n_t_proposals}), and thus start executing \textsc{Quad} at $\text{GST} + \delta$  (line~\ref{line:propose_quad}).
Nevertheless, it is easy to show that \textsc{Quad} ensures the stated complexity even if all correct processes start executing the algorithm by time $\text{GST} + \delta$.
Not only that, even if correct processes do not start executing \textsc{Quad} within $\delta$ time after GST, but they all start executing the algorithm within $\delta$ time from each other (after GST), the message complexity remains quadratic and the latency remains linear (measured from the time the first correct process starts executing \textsc{Quad}).

\paragraph{Correctness \& complexity.}

First, we prove the correctness.

\begin{theorem}
\Cref{algorithm:interactive_consistency} is correct.
\end{theorem}
\begin{proof}
\emph{Agreement} follows directly from the fact that \textsc{Quad} satisfies \emph{Agreement}.
\emph{Termination} follows from (1) \emph{Termination} of \textsc{Quad}, and (2) the fact that, eventually, all correct processes receive $n - t$ \textsc{proposal} messages (as there are at least $n - t$ correct processes).

We now prove that \Cref{algorithm:interactive_consistency} satisfies \emph{Vector Validity}.
Let a correct process $P$ decide $\mathit{vector}' \in \mathcal{I}_{n - t}$ from vector consensus (line~\ref{line:decide}).
Hence, $P$ has decided $(\mathit{vector}', \Sigma')$ from \textsc{Quad}, where (1) $\Sigma'$ is some proof, and (2) $\mathsf{verify}(\mathit{vector}', \Sigma') = \mathit{true}$ (due to the specification of \textsc{Quad}).
Furthermore, if there exists a process-proposal pair $(P, v \in \mathcal{V}_I)$ in $\mathit{vector}'$, where $P$ is a correct process, a properly signed \textsc{proposal} message belongs to $\Sigma'$.
As correct processes send \textsc{proposal} messages only for their proposals (line~\ref{line:broadcast_proposal}), $v$ was indeed proposed by $P$.
Thus, the theorem.
\end{proof}

Finally, we prove the complexity.

\begin{theorem}
The message complexity of \Cref{algorithm:interactive_consistency} is $O(n^2)$.
\end{theorem}
\begin{proof}
The message complexity of the specific instance of \textsc{Quad} utilized in \Cref{algorithm:interactive_consistency} is $O(n^2)$.
Additionally, correct processes exchange $O(n^2)$ \textsc{proposal} messages.
Thus, the message complexity is $O(n^2) + O(n^2) = O(n^2)$.
\end{proof}

\subsection{Non-Authenticated Implementation: Pseudocode \& Proofs} \label{subsection:nonauthenticated_vector}

We now present a non-authenticated implementation (\Cref{algorithm:non_authenticated}) of vector consensus.
The design of \Cref{algorithm:non_authenticated} follows the reduction from binary consensus to multivalued consensus (e.g.,~\cite{CGL18}).
Namely, we use the following two building blocks in \Cref{algorithm:non_authenticated}:
\begin{compactenum}
    \item Byzantine Reliable Broadcast~\cite{cachin2011introduction, B87}:
    This primitive allows processes to disseminate information in a reliable manner.
    Formally, Byzantine reliable broadcast exposes the following interface: (1) \textbf{request} $\mathsf{broadcast}(m)$, and (2) \textbf{indication} $\mathsf{deliver}(m')$.
    The primitive satisfies the following properties:
    \begin{compactitem}
        \item \emph{Validity:} If a correct process $P$ broadcasts a message $m$, $P$ eventually delivers $m$.
        
        \item \emph{Consistency:} No two correct processes deliver different messages.
        
        \item \emph{Integrity:} Every correct process delivers at most one message.
        Moreover, if a correct process delivers a message $m$ from a process $P$ and $P$ is correct, then $P$ broadcast $m$.
        
        \item \emph{Totality:} If a correct process delivers a message, every correct process delivers a message.
    \end{compactitem}
    In \Cref{algorithm:non_authenticated}, we use a non-authenticated implementation~\cite{B87} of the Byzantine reliable broadcast primitive.
    
    \item Binary DBFT~\cite{CGL18}, a non-authenticated algorithm which solves the Byzantine consensus problem with $\emph{Strong Validity}$.
\end{compactenum}

Let us briefly explain how \Cref{algorithm:non_authenticated} works; we focus on a correct process $P_i$.
First, $P_i$ reliably broadcasts its proposal (line~\ref{line:broadcast_proposal_nowic}).
Once $P_i$ delivers a proposal of some process $P_j$ (line~\ref{line:receive_proposal_nowic}), $P_i$ proposes $1$ to the corresponding DBFT instance (line~\ref{line:propose_dbft_nowic}).
Eventually, $n - t$ DBFT instances decide $1$ (line~\ref{line:dbft_decided_nowic}).
Once that happens, $P_i$ proposes $0$ to all DBFT instances to which $P_i$ has not proposed (line~\ref{line:propose_0_dbft_nowic}).
When all DBFT instances have decided (line~\ref{line:nowic_decide_rule}), $P_i$ decides an input configuration associated with the first $n - t$ processes whose DBFT instances decided $1$ (constructed at line~\ref{line:input_configuration_nowic}).

\begin{algorithm*}
\caption{Non-Authenticated Vector Consensus: Pseudocode (for process $P_i$)}
\label{algorithm:non_authenticated}
\footnotesize
\begin{algorithmic} [1]
\State \textbf{Uses:}
\State \hskip2em Non-Authenticated Byzantine Reliable Broadcast~\cite{B87}, \textbf{instance} $\mathit{brb}$
\State \hskip2em Binary DBFT~\cite{CGL18}, \textbf{instances} $\mathit{dbft}[1]$, ..., $\mathit{dbft}[n]$ \BlueComment{one instance of the binary DBFT algorithm per process}

\medskip
\State \textbf{upon} $\mathsf{init}$:
\State \hskip2em $\mathsf{Map}(\mathsf{Process} \to \mathcal{V}_I)$ $\mathit{proposals}_i \gets \text{empty}$ \BlueComment{received proposals}

\State \hskip2em $\mathsf{Boolean}$ $\mathit{dbft\_proposing}_i = \mathit{true}$ \BlueComment{is $P_i$ still proposing $1$s to the DBFT instances}
\State \hskip2em $\mathsf{Map}(\mathsf{Process} \to \mathsf{Boolean})$ $\mathit{dbft\_proposed}_i \gets \{\mathit{false}, \text{for every } \mathsf{Process}\}$ 
\State \hskip2em $\mathsf{Integer}$ $\mathit{dbft\_decisions}_i \gets 0$ \BlueComment{the number of the DBFT instances which have decided}

\medskip
\State \textbf{upon} $\mathsf{propose}(v \in \mathcal{V}_I)$: \label{line:propose_nowic}
\State \hskip2em \textbf{invoke} $\mathit{brb}.\mathsf{broadcast}\big( \langle \textsc{proposal}, v \rangle \big)$ \BlueComment{broadcast a proposal} \label{line:broadcast_proposal_nowic}

\medskip
\State \textbf{upon} reception of $\mathsf{Message}$ $m = \langle \textsc{proposal}, v_j \in \mathcal{V}_I \rangle$ from process $P_j$: \label{line:receive_proposal_nowic} \BlueComment{delivered from $\mathit{brb}$}
\State \hskip2em $\mathit{proposals}_i[P_j] \gets v_j$

\State \hskip2em \textbf{if} $\mathit{dbft\_proposing}_i = \mathit{true}$:
\State \hskip4em $\mathit{dbft\_proposed}_i[P_j] \gets \mathit{true}$
\State \hskip4em \textbf{invoke} $\mathit{dbft}[j].\mathsf{propose}(1)$ \label{line:propose_dbft_nowic}

\medskip
\State \textbf{upon} $n - t$ DBFT instances have decided 1 (for the first time): \label{line:dbft_decided_nowic}
\State \hskip2em $\mathit{dbft\_proposing}_i \gets \mathit{false}$
\State \hskip2em \textbf{for} every $\mathsf{Process}$ $P_j$ such that $\mathit{dbft\_proposed}_i[P_j] = \mathit{false}$:
\State \hskip4em $\mathit{dbft\_proposed}_i[P_j] \gets \mathit{true}$
\State \hskip4em \textbf{invoke} $\mathit{dbft}[j].\mathsf{propose}(0)$ \label{line:propose_0_dbft_nowic}

\medskip
\State \textbf{upon} all DBFT instances decided, and, for the first $n - t$ processes $P_j$ such that $\mathit{dbft}[j]$ decided $1$, $\mathit{proposals}_i[P_j] \neq \bot$: \label{line:nowic_decide_rule}
\State \hskip2em $\mathsf{Input\_Configuration}$ $\mathit{vector} \gets $ input configuration with $n - t$ process-proposal pairs corresponding to the first $n - t$ DBFT instances which decided 1 \label{line:input_configuration_nowic} 
\State \hskip2em \textbf{trigger} $\mathsf{decide}(\mathit{vector})$ \label{line:decide_nowic}
\end{algorithmic}
\end{algorithm*}

\begin{theorem}
\Cref{algorithm:non_authenticated} is correct.
\end{theorem}
\begin{proof}

We start by proving \emph{Termination} of \Cref{algorithm:non_authenticated}.
Eventually, at least $n - t$ DBFT instances decide $1$ due to the fact that (1) no correct process proposes $0$ to any DBFT instance unless $n - t$ DBFT instances have decided $1$ (line~\ref{line:dbft_decided_nowic}), and (2) all correct processes eventually propose $1$ to the DBFT instances which correspond to the correct processes (unless $n - t$ DBFT instances have already decided $1$).
When $n - t$ DBFT instances decide $1$ (line~\ref{line:dbft_decided_nowic}), each correct process proposes to all instances to which it has not yet proposed (line~\ref{line:propose_0_dbft_nowic}).
Hence, eventually all DBFT instances decide, and (at least) $n - t$ DBFT instances decide $1$.
Therefore, the rule at line~\ref{line:nowic_decide_rule} eventually activates at every correct process, which implies that every correct process eventually decides (line~\ref{line:decide_nowic}).

Next, we prove \emph{Vector Validity}.
If a correct process $P$ decides an input configuration with a process-proposal pair $(Q, v)$, $P$ has delivered a \textsc{proposal} message from $Q$ (line~\ref{line:receive_proposal_nowic}).
If $Q$ is correct, due to integrity of the reliable broadcast primitive, $Q$'s proposal was indeed $v$.

Finally, \emph{Agreement} follows from (1) \emph{Agreement} of DBFT, and (2) consistency of the reliable broadcast primitive.
Therefore, \Cref{algorithm:non_authenticated} is correct.
\end{proof}

The main downside of \Cref{algorithm:non_authenticated} is that its message complexity is $O(n^4)$.
Therefore, non-authenticated version of \general has $O(n^4)$ message complexity, which is not optimal according to our lower bound (\Cref{subsection:lower_bound}).
Interestingly, as \Cref{algorithm:non_authenticated} shows that vector consensus can be solved using consensus with \emph{Strong Validity}, \emph{Strong Validity} is ``another'' strongest validity property: a solution to consensus with \emph{Strong Validity} yields a solution to any non-trivial consensus variant (although with additional cost).

\subsection{Implementation with $O(n^2 \log n)$ Communication: Pseudocode \& Proofs} \label{subsection:better_communication}

In this subsection, we give an implementation of vector consensus with $O(n^2 \log n)$ communication complexity, which comes within a logarithmic factor of the lower bound (\Cref{subsection:lower_bound}).
This implementation represents a near-linear communication improvement over \Cref{algorithm:interactive_consistency} (\Cref{subsection:general_algorithm}), which achieves $O(n^3)$ communication complexity.
We note that the following solution is highly impractical due to its exponential latency.
However, our solution does represent a step towards closing the existing gap in the communication complexity of consensus with non-trivial (and solvable) validity properties.

\paragraph{Threshold signatures.}
For this implementation, we assume a $(k, n)$-threshold signature scheme~\cite{libert2014born}, where $k = n - t$.
In a threshold signature scheme, each process holds a distinct private key, and there exists a single public key.
Each process $P_i$ can use its private key to produce a (partial) signature of a message $m$.
Moreover, a signature can be verified by other processes.
Finally, a set of signatures for a message $m$ from $k$ (the threshold) distinct processes can be combined into a single threshold signature for $m$, which proves that $k$ processes have signed $m$.

\paragraph{Hash functions.}
We denote by $\mathsf{hash}(\cdot)$ a collision-resistant hash function.

\subsubsection{Vector Dissemination} \label{subsubsection:vector_dissemination}
First, we formally define the \emph{vector dissemination} problem, which plays the crucial role in our vector consensus algorithm with improved communication complexity.
In this problem, each correct process \emph{disseminates} a vector of exactly $n - t$ values, and all correct processes eventually acquire (1) a hash of some disseminated vector, and (2) a threshold signature for that hash.
Formally, the vector dissemination problem exposes the following interface:
\begin{compactitem}
    \item \textbf{request} $\mathsf{disseminate}(\mathsf{Vector} \text{ } \mathit{vec})$: a process disseminates a vector $\mathit{vec}$. 
    
    \item \textbf{indication} $\mathsf{acquire}(\mathsf{Hash\_Value} \text{ } H, \mathsf{T\_Signature} \text{ } \mathit{tsig})$: a process acquires a hash value $H'$ and a threshold signature $\mathit{tsig}$.
\end{compactitem}
The following properties are required:
\begin{compactitem}
    \item \emph{Termination:} Every correct process eventually acquires a hash-signature pair.

    \item \emph{Integrity:} If a correct process acquires a hash-signature pair $(H, \mathit{tsig})$, then $\mathsf{verify\_sig}(H, \mathit{tsig}) = \mathit{true}$.
    
    
    \item \emph{Redundancy:} Let a correct process cache (i.e., store in its local memory) a threshold signature $\mathit{tsig}$ for some hash value $H$.
    Then, (at least) $t + 1$ correct processes have cached a vector $\mathit{vec}$ such that $\mathsf{hash}(\mathit{vec}) = H$.
\end{compactitem}

\paragraph{Slow broadcast.}
In order to solve the vector dissemination problem, we present a simple algorithm (\Cref{algorithm:slow_1}) which implements \emph{slow broadcast}.
In slow broadcast, each process disseminates its vector in ``one-by-one'' fashion, with a ``waiting step'' between any two sending events.
Specifically, process $P_1$ broadcasts its vector by (1) sending the vector to $P_1$ (line~\ref{line:slow_send}), and then waiting $\delta$ time (line~\ref{line:slow_wait}), (2) sending the vector to $P_2$ (line~\ref{line:slow_send}), and then waiting $\delta$ time (line~\ref{line:slow_wait}), etc.
Process $P_2$ broadcasts its vector in the same manner, but it waits $\delta \cdot n$ time (line~\ref{line:slow_wait}).
Crucially, if the system is synchronous, the waiting time of $P_2$ is (roughly) sufficient for $P_1$ to \emph{completely} disseminate its vector.
This holds for any two processes $P_i$ and $P_j$ with $i < j$.

\begin{algorithm}[h]
\caption{Slow Broadcast: Pseudocode (for process $P_i$)}
\label{algorithm:slow_1}
\footnotesize
\begin{algorithmic} [1]

\State \textbf{upon} $\mathsf{broadcast}( \mathsf{Vector}~\mathit{vec} )$:
\State \hskip2em \textbf{for each} $\mathsf{Process}$ $P_j$:
\State \hskip4em \textbf{send} $\langle \textsc{slow\_broadcast}, \mathit{vec} \rangle~\textbf{to}~P_j$ \label{line:slow_send}
\State \hskip4em \textbf{wait for} $\delta \cdot n^{(i-1)}$ time \label{line:slow_wait}

\medskip
\State \textbf{upon} reception of $\langle \textsc{slow\_broadcast}, \mathsf{Vector} \text{ } \mathit{vec}' \rangle$ from process $P_j$:
\State \hskip2em \textbf{trigger} $\mathsf{deliver}(\mathit{vec}', P_j)$

\end{algorithmic}
\end{algorithm}









\paragraph{Algorithm description.}
Our solution is given in \Cref{algorithm:prep}.
We explain it from the perspective of a correct process $P_i$.

When $P_i$ starts disseminating its vector $\mathit{vec}$ (line~\ref{line:disseminate_vector}), $P_i$ stores its hash (line~\ref{line:compute_hash}) and slow-broadcasts the vector (line~\ref{line:slow_broadcast}).
Once $P_i$ receives \textsc{stored} messages from $n - t$ distinct processes (line~\ref{line:receive_acknowledgements_vector_dissemination}), $P_i$ combines the received partial signatures into a threshold signature (line~\ref{line:create_storage_proof}).
Then, $P_i$ broadcasts (using the best-effort broadcast primitive) the signature (line~\ref{line:broadcast_storage_proof}).

Whenever $P_i$ receives a threshold signature (line~\ref{line:receive_storage_proof}), $P_i$ checks whether the signature is valid (line~\ref{line:check_storage_proof}).
If it is, $P_i$ rebroadcasts the signature (line~\ref{line:forward_storage_proof}), acquires a hash value and the signature (line~\ref{line:obtain_hash}), and stops participating (i.e., sending and processing messages) in vector dissemination (line~\ref{line:stop_participating}).
Observe that, once $P_i$ stops participating in vector dissemination (line~\ref{line:stop_participating}), it stops participating in slow broadcast, as well.

\begin{algorithm*} [h]
\caption{Vector Dissemination: Pseudocode (for process $P_i$)}
\label{algorithm:prep}
\footnotesize
\begin{algorithmic} [1]
\State \textbf{Uses:}
\State \hskip2em Best-Effort Broadcast~\cite{cachin2011introduction}, \textbf{instance} $\mathit{beb}$ \BlueComment{broadcast with no guarantees if the sender is faulty}
\State \hskip2em Slow Broadcast, \textbf{instance} $\mathit{slow}$ \BlueComment{see \Cref{algorithm:slow_1}}

\medskip
\State \textbf{upon} $\mathsf{init}$:
\State \hskip2em $\mathsf{Hash\_Value} \text{ } H_i \gets \bot$ \BlueComment{hash of the message $P_i$ slow-broadcasts}
\State \hskip2em $\mathsf{Map}(\mathsf{Hash\_Value} \to \mathsf{Vector})$ $\mathit{vectors}_i \gets \text{empty}$ \BlueComment{received vectors}
\State \hskip2em $\mathsf{Set}(\mathsf{Process})$ $\mathit{disseminated}_i \gets \text{empty}$ \BlueComment{processes who have disseminated a vector}

\medskip
\State \textbf{upon} $\mathsf{disseminate}( \mathsf{Vector}~\mathit{vec})$: \label{line:disseminate_vector}
\State \hskip2em $H_i \gets \mathsf{hash}(\mathit{vec})$ \label{line:compute_hash}
\State \hskip2em \textbf{invoke} $\mathit{slow}.\mathsf{broadcast}(\mathit{vec})$ \label{line:slow_broadcast}

\medskip
\State \textbf{upon} $\mathit{slow}.\mathsf{deliver}( \mathsf{Vector}~\mathit{vec}', \mathsf{Process} \text{ } P_j)$: \label{line:deliver_slow_bcast}
\State \hskip2em \textbf{if} $P_j \notin \mathit{disseminated}_i$: \label{line:verify_vector_signature}
\State \hskip4em $\mathit{disseminated}_i \gets \mathit{disseminated}_i \cup \{P_j\}$
\State \hskip4em $\mathit{vectors}_i[\mathsf{hash}(\mathit{vec}')] \gets \mathit{vec}'$ \label{line:store_preimage} \BlueComment{cache $\mathit{vec}'$}
\State \hskip4em \textbf{send} $\langle \textsc{stored}, \mathsf{hash}(\mathit{vec}'), \mathsf{partially\_sign}_i\big( \mathsf{hash}(\mathit{vec'}) \big) \rangle~\textbf{to}~P_j$ \label{line:send_stored} \BlueComment{acknowledge the reception by sending a partial signature to $P_j$}

\medskip
\State \textcolor{blue}{\(\triangleright\) acknowledgements are received}
\State \textbf{upon} reception of $\mathsf{Message}$ $m_j = \langle \textsc{stored}, \mathsf{Hash\_Value} \text{ } H', \mathsf{P\_Signature} \text{ } \mathit{psig} \rangle$ such that $H' = H_i$ from $n - t$ distinct processes: \label{line:receive_acknowledgements_vector_dissemination} \label{line:received_n-t_stored}
\State \hskip2em $\mathsf{T\_Signature}~\mathit{tsig} \gets \mathsf{combine}\big(\{ \mathit{psig} \,|\, \mathit{psig} \text{ is a signature received in the \textsc{stored} messages} \}\big)$ \label{line:create_storage_proof}
\State \hskip2em \textbf{invoke} $\mathit{beb}.\mathsf{broadcast}\big( \langle \textsc{confirm}, H_i, \mathit{tsig}\rangle \big)$ \label{line:bcast_storage_proof} \BlueComment{disseminate the threshold signature} \label{line:broadcast_storage_proof}

\medskip
\State \textcolor{blue}{\(\triangleright\) a threshold signature is received}
\State \textbf{upon} reception of $\mathsf{Message}$ $m = \langle \textsc{confirm}, \mathsf{Hash\_Value} \text{ } H', \mathsf{T\_Signature} \text{ } \mathit{tsig}'\rangle$: \label{line:receive_storage_proof}
\State \hskip2em \textbf{if} $\mathit{tsig}'$ is a valid $(n-t)$-threshold signature for $H'$: \label{line:check_storage_proof} \BlueComment{check that the threshold signature is valid}
\State \hskip4em \textbf{invoke} $\mathit{beb}.\mathsf{broadcast}\big( \langle \textsc{confirm}, H', \mathit{tsig}'\rangle \big)$ \label{line:forward_storage_proof} \BlueComment{rebroadcast the threshold signature}
\State \hskip4em \textbf{trigger} $\mathsf{acquire}(H', \mathit{tsig}')$ \label{line:obtain_hash}
\State \hskip4em \textbf{stop participating} in vector dissemination (and slow broadcast) \label{line:stop_participating}

\end{algorithmic}
\end{algorithm*}

\paragraph{Correctness \& complexity.}
We start by proving redundancy.

\begin{lemma} \label{lemma:redundancy}
\Cref{algorithm:prep} satisfies redundancy.
\end{lemma}
\begin{proof}
Let a correct process cache a threshold signature $\mathit{tsig}'$ for some hash value $H'$.
Hence, $n - t$ processes have (partially) signed $H'$.
Among these $n - t$ processes, at least $t + 1$ are correct (as $n > 3t$).
Before signing $H'$ (line~\ref{line:send_stored}), all these correct processes cache a vector $\mathit{vec}'$ (line~\ref{line:store_preimage}), where $\mathsf{hash}(\mathit{vec}') = H'$.
\end{proof}



The following lemma proves that, if a correct process $P_i$ starts the dissemination of its vector at time $t_i$, then every correct process acquires a hash value and a threshold signature by time $\max(\text{GST}, t_i) + \delta \cdot n^i + 3\delta$.
We emphasize that the $\max(\text{GST}, t_i) + \delta \cdot n^i + 3\delta$ time is not tight; we choose it for simplicity of presentation.

\begin{lemma} \label{lemma:timely_termination}
If a correct process $P_i$ starts the dissemination of its vector at time $t_i$, every correct process acquires a hash value and a threshold signature by time $\max(\text{GST}, t_i) + \delta \cdot n^i + 3\delta$.
\end{lemma}
\begin{proof}
We separate the proof into two cases:
\begin{compactitem}
    \item There exists a correct process which acquires a hash-signature pair by time $\max(\text{GST}, t_i) + \delta \cdot n^i + 2\delta$.
    In this case, the statement of the lemma holds as every correct process acquires a pair by time $\max(\text{GST}, t_i) + \delta \cdot n^i + 3\delta$ due to the ``rebroadcasting step'' (line~\ref{line:forward_storage_proof}).
    
    \item There does not exist a correct process which acquires a hash-signature pair by time $T = \max(\text{GST}, t_i) + \delta \cdot n^i + 2\delta$.
    Hence, no process stops participating in vector dissemination by time $T$, i.e., no process executes line~\ref{line:stop_participating} by time $T$.
    Every correct process receives a \textsc{slow\_broadcast} message from process $P_i$ by time $\max(\text{GST}, t_i) + \delta \cdot n^i + \delta$.

    Thus, by time $\max(\text{GST}, t_i) + \delta \cdot n^i + 2\delta$, $P_i$ receives $n - t$ partial signatures (line~\ref{line:receive_acknowledgements_vector_dissemination}).
    Finally, by time $\max(\text{GST}, t_i) + \delta \cdot n^i + 3\delta$, every correct process receives a \textsc{confirm} message from $P_i$ (line~\ref{line:receive_storage_proof}), and acquires a hash-signature pair (line~\ref{line:obtain_hash}).
    In this case, the statement of the lemma holds.
\end{compactitem}
The lemma holds as its statement is satisfies in both cases.
\end{proof}

The next lemma proves that \Cref{algorithm:prep} satisfies termination.

\begin{lemma} \label{lemma:termination_prep}
\Cref{algorithm:prep} satisfies termination.
\end{lemma}
\begin{proof}
Follows directly from \Cref{lemma:timely_termination}.
\end{proof}

Next, we prove integrity.

\begin{lemma} \label{lemma:integrity}
\Cref{algorithm:prep} satisfies integrity.
\end{lemma}
\begin{proof}
Follows from the check at line~\ref{line:check_storage_proof}.
\end{proof}

Therefore, \Cref{algorithm:prep} solves the vector dissemination problem.

\begin{theorem}
\Cref{algorithm:prep} is correct.
\end{theorem}

Lastly, we prove that the communication complexity of \Cref{algorithm:prep} is $O(n^2)$.
Recall that the communication complexity denotes the number of words sent by correct processes at and after GST, where a word consists of a constant number of values, hashes and signatures.

\begin{theorem}\label{theorem:quadratic_dissemination}
Let no correct process start the dissemination of its vector after time $\text{GST} +  \delta$.
Then, the communication complexity of \Cref{algorithm:prep} is $O(n^2)$. 
\end{theorem}
\begin{proof}
Let $i$ be the minimum index such that (1) process $P_i$ is correct, and (2) $P_i$ sends a \textsc{slow\_broadcast} message at some time $\geq \text{GST}$.
If $i$ does not exist, the lemma trivially holds.

Let $t_i$ denote the time at which $P_i$ starts the dissemination of its vector (line~\ref{line:disseminate_vector}).
By assumption, $t_i \leq \text{GST} + \delta$.
Every correct process acquires a hash-signature pair by time $\max(\text{GST}, t_i) + \delta \cdot n^i + 3\delta$ (by \Cref{lemma:timely_termination}).
Thus, as $t_i \leq \text{GST} + \delta$, every correct process acquires a hash-signature pair by time $\text{GST} + \delta \cdot n^i + 4\delta$.
Moreover, by time $\text{GST} + \delta \cdot n^i + 4\delta$, all correct processes stop sending \textsc{slow\_broadcast} messages (due to line~\ref{line:stop_participating}).

Let $P_j$ be a correct process such that $j > i$.
Due to the slow broadcast primitive (\Cref{algorithm:slow_1}), $P_j$ has a ``waiting step'' of (at least) $\delta \cdot n^i $ time (after GST).
Therefore, during the $[\text{GST}, \text{GST} + \delta \cdot n^i + 4\delta]$ period, $P_j$ can send only $O(1)$ \textsc{slow\_broadcast} messages.
Thus, at most one correct process (i.e., $P_i$) sends more than $O(1)$ \textsc{slow\_broadcast} messages during the $[\text{GST}, \text{GST} + \delta \cdot n^i + 4\delta]$ period; that process sends at most $n$ \textsc{slow\_broadcast} messages.
As each message is of size $O(n)$ (since it carries a vector of $n - t$ values), the communication complexity of \Cref{algorithm:prep} is $O(n) \cdot O(1) \cdot O(n) + 1 \cdot O(n) \cdot O(n) = O(n^2)$.
\end{proof}

\subsubsection{Vector Consensus with $O(n^2\log n)$ Communication Complexity}
Finally, we are ready to present our vector consensus algorithm (\Cref{algorithm:better_interactive_consistency}) with subcubic communication complexity.
Our algorithm consists of three building blocks: (1) vector dissemination (\Cref{subsubsection:vector_dissemination}), (2) \textsc{Quad} (\Cref{subsubsection:interactive_consistency}), and (3) \textsc{add}~\cite{das2021asynchronous}, an algorithm for asynchronous data dissemination.
In \Cref{algorithm:better_interactive_consistency}, we rely on a specific instance of \textsc{Quad} in which (1) each proposal value is a hash value, and (2) given a hash value $H$ and a \textsc{Quad}'s proof $\Sigma$, $\mathsf{verify}(H, \Sigma) = \mathit{true}$ if and only if $\Sigma$ is a valid $(n - t)$-threshold signature for $H$.
Below, we briefly explain \textsc{add}.

\paragraph{\textsc{add}.}
This algorithm solves the \emph{data dissemination}~\cite{das2021asynchronous} problem defined in the following way.
Let $M$ be a data blob which is an input of (at least) $t + 1$ correct processes; other correct processes have $\bot$ as their input.
The data dissemination problem requires every correct process to eventually output $M$, and no other message.
The key feature of \textsc{add} is that it solves the problem with $O(n^2 \log n)$ communication complexity.
(For the full details on \textsc{add}, see~\cite{das2021asynchronous}.)

\paragraph{Algorithm description.}
We give the description of \Cref{algorithm:better_interactive_consistency} from the perspective of a correct process $P_i$.
When $P_i$ proposes its value (line~\ref{line:propose_log}), it disseminates the value (using the best-effort broadcast primitive) to all processes (line~\ref{line:bcast_proposal}).
Once $P_i$ receives proposals of $n - t$ distinct processes (line~\ref{line:received_n-t_proposals_before_dissemination}), it constructs an input configuration (line~\ref{line:construct_vector_before_dissemination}), and starts disseminating it (line~\ref{line:disseminate}).\footnote{Recall that this input configuration is actually a vector of $n - t$ values.}

When $P_i$ acquires a hash value $H$ and a threshold signature $\mathit{tsig}$ (line~\ref{line:prep_store}), $P_i$ proposes $(H, \mathit{tsig})$ to \textsc{Quad} (line~\ref{line:quad_propose}).
Observe that $\mathsf{verify}(H, \mathit{tsig}) = \mathit{true}$ (due to the integrity property of vector dissemination).
Once $P_i$ decides from \textsc{Quad} (line~\ref{line:quad_decide}), it starts \textsc{add} (line~\ref{line:add_input}).
Specifically, $P_i$ checks whether it has cached an input configuration whose hash value is $H'$ (line~\ref{line:check_cache}).
If so, $P_i$ inputs the input configuration to \textsc{add}; otherwise, $P_i$ inputs $\bot$.
Once $P_i$ outputs an input configuration from \textsc{add} (line~\ref{line:add_output}), it decides it (line~\ref{line:decide_better}).

\begin{algorithm*}
\caption{$O(n^2 \log n)$ Vector Consensus: Pseudocode (for process $P_i$)}
\label{algorithm:better_interactive_consistency}
\footnotesize
\begin{algorithmic} [1]
\State \textbf{Uses:}
\State \hskip2em Best-Effort Broadcast~\cite{cachin2011introduction}, \textbf{instance} $\mathit{beb}$ \BlueComment{broadcast with no guarantees if the sender is faulty}
\State \hskip2em Vector Dissemination, \textbf{instance} $\mathit{disseminator}$ \BlueComment{see \Cref{algorithm:prep}}
\State \hskip2em \textsc{Quad}~\cite{civit2022byzantine}, \textbf{instance} $\mathit{quad}$
\State \hskip2em \textsc{add}~\cite{das2021asynchronous}, \textbf{instance} $\mathit{add}$

\medskip
\State \textbf{upon} $\mathsf{init}$:
\State \hskip2em $\mathsf{Integer}$ $\mathit{received\_proposals}_i \gets 0$ \BlueComment{the number of received proposals}
\State \hskip2em $\mathsf{Map}(\mathsf{Process} \to \mathcal{V}_I)$ $\mathit{proposals}_i \gets \text{empty}$ \BlueComment{received proposals}
\State \hskip2em $\mathsf{Map}(\mathsf{Process} \to \mathsf{Message})$ $\mathit{messages}_i \gets \text{empty}$ \BlueComment{received \textsc{proposal} messages}

\medskip
\State \textbf{upon} $\mathsf{propose}(v \in \mathcal{V}_I)$: \label{line:propose_log}
\State \hskip2em \textbf{invoke} $\mathit{beb}.\mathsf{broadcast}\big( \langle \textsc{proposal}, v \rangle_{\sigma_i} \big)$ \label{line:bcast_proposal} \BlueComment{broadcast a signed proposal}

\medskip
\State \textbf{upon} reception of $\mathsf{Message}$ $m = \langle \textsc{proposal}, v_j \in \mathcal{V}_I \rangle_{\sigma_j}$ from process $P_j$ and $\mathit{received\_proposals}_i < n - t$:
\State \hskip2em $\mathit{received\_proposals}_i \gets \mathit{received\_proposals}_i + 1$
\State \hskip2em $\mathit{proposals}_i[P_j] \gets v_j$
\State \hskip2em $\mathit{messages}_i[P_j] \gets m$

\State \hskip2em \textbf{if} $\mathit{received\_proposals}_i = n - t$: \label{line:received_n-t_proposals_before_dissemination} \BlueComment{received $n - t$ proposals; can start disseminating}
\State \hskip4em $\mathsf{Input\_Configuration}$ $\mathit{vector} \gets $ input configuration constructed from $\mathit{proposals}_i$ \label{line:construct_vector_before_dissemination}
\State \hskip4em \textbf{invoke} $\mathit{disseminator}.\mathsf{disseminate}(\mathit{vector})$ \label{line:disseminate}

\medskip
\State \textbf{upon} $\mathit{disseminator}.\mathsf{acquire}\big( (\mathsf{Hash\_Value}~\mathit{H}, \mathsf{T\_Signature}~\mathit{tsig}) \big)$: \label{line:prep_store}
\State \hskip2em \textbf{if} have not yet proposed to \textsc{Quad}:
\State \hskip4em \textbf{invoke} $\mathit{quad}.\mathsf{propose}\big( (\mathit{H}, \mathit{tsig}) \big)$ \label{line:quad_propose}

\medskip
\State \textbf{upon} $\mathit{quad}.\mathsf{decide}\big( (\mathsf{Hash\_Value}~\mathit{H}', \mathsf{T\_Signature}~\mathit{tsig}') \big)$: \label{line:quad_decide}
\State \hskip2em $\mathsf{Input\_Configuration}$ $\mathit{vector}' \gets$ a cached vector whose hash value is $H'$ \BlueComment{can be $\bot$} \label{line:check_cache}
\State \hskip2em \textbf{invoke} $\mathit{add}.\mathsf{input}( \mathit{vector}' )$ \label{line:add_input}

\medskip
\State \textbf{upon} $\mathit{add}.\mathsf{output}\big(\mathsf{Input\_Configuration} \text{ } \mathit{vector}'' \big)$: \label{line:add_output}
\State \hskip2em \textbf{trigger} $\mathsf{decide}(\mathit{vector}'')$ \label{line:decide_better}

\end{algorithmic}
\end{algorithm*}

\paragraph{Correctness \& complexity.}

We start by proving that (1) all correct processes eventually output a non-$\bot$ value from \textsc{add}, and (2) no two correct processes output different values from \textsc{add}.

\begin{lemma} \label{lemma:eventually_add}
The following holds:
\begin{compactitem}
    \item Every correct process eventually outputs a non-$\bot$ value from \textsc{add} (line~\ref{line:add_output}); moreover, the output value was an input (to \textsc{add}) of a correct process.

    \item No two correct processes output different input configurations from \textsc{add} (line~\ref{line:add_output}).
\end{compactitem}
\end{lemma}
\begin{proof}
First, every correct process broadcasts its proposal (line \ref{line:bcast_proposal}).
Thus, every correct process eventually receives $n-t$ proposals (line \ref{line:received_n-t_proposals_before_dissemination}), and starts the dissemination of an input configuration (line~\ref{line:disseminate}).
Due to the termination property of vector dissemination (\Cref{lemma:termination_prep}), every correct process acquires a hash-signature pair (line~\ref{line:prep_store}).
Hence, every correct process eventually proposes to \textsc{Quad} (line~\ref{line:quad_propose}).
Due to \emph{Termination} of \textsc{Quad}, every correct process eventually decides from \textsc{Quad} (line~\ref{line:quad_decide}), and starts executing \textsc{add} (line~\ref{line:add_input}).

As the pair decided from \textsc{Quad} (recall that \textsc{Quad} satisfies \emph{Agreement}) includes a threshold signature, at least $t + 1$ correct processes have cached an input configuration whose hash value is decided from \textsc{Quad} (by redundancy of vector dissemination).
Therefore, all of these correct processes input the same non-$\bot$ value to \textsc{add} (line~\ref{line:add_input}); let that value be $\mathit{vec}$.
Moreover, no correct process inputs a different non-$\bot$ value to \textsc{add}.
Therefore, the conditions required by \textsc{add} are met, which implies that all correct processes eventually output $\mathit{vec} \neq \bot$ from \textsc{add} (line~\ref{line:add_output}).
\end{proof}

The following theorem proves that \Cref{algorithm:better_interactive_consistency} is correct.

\begin{theorem}
\Cref{algorithm:better_interactive_consistency} is correct.
\end{theorem}
\begin{proof}
\emph{Agreement} and \emph{Termination} follow from \Cref{lemma:eventually_add}.

It is left to prove \emph{Vector Validity}.
Let $\mathit{vec}'$ be an input configuration of $n - t$ proposals decided by a correct process (line~\ref{line:decide_better}).
Hence, $\mathit{vec}'$ is an input (to \textsc{add}) of a correct process (by \Cref{lemma:eventually_add}), which implies that some correct process has previously cached $\mathit{vec}'$.
Before a correct process caches a vector (\Cref{algorithm:prep}), it verifies that it is associated with corresponding \textsc{proposal} messages; we omit this check for brevity.
As correct processes only send \textsc{proposal} messages for their proposals (line~\ref{line:bcast_proposal}), \emph{Vector Validity} is satisfied.
\end{proof}

Lastly, we show the communication complexity of \Cref{algorithm:better_interactive_consistency}.

\begin{theorem}
The communication complexity of \Cref{algorithm:better_interactive_consistency} is $O(n^2 \log n)$.
\end{theorem}
\begin{proof}
The communication complexity of a single best-effort broadcast instance is $O(n)$.
Every correct process starts the dissemination of its vector by time $\text{GST} + \delta$ (as every correct process receives $n - t$ proposals by this time).
Thus, the communication complexity of vector dissemination is $O(n^2)$ (by \Cref{theorem:quadratic_dissemination}).
The communication complexity of \textsc{Quad} is $O(n^2)$.
Moreover, the communication complexity of \textsc{add} is $O(n^2 \log n)$ (see~\cite{das2021asynchronous}).
As \Cref{algorithm:better_interactive_consistency} is a composition of the aforementioned building blocks, its communication complexity is $n \cdot O(n) + O(n^2) + O(n^2) + O(n^2 \log n) = O(n^2 \log n)$.
\end{proof}
\section{Extended Formalism} \label{section:extended_formalim_appendix}

In this section, we give intuition behind an extension of our formalism which is suitable for the analysis of blockchain-specific validity properties, such as \emph{External Validity}~\cite{Cachin2001,BKM19,yin2019hotstuff}.
\emph{External Validity} stipulates that any decided value must satisfy a predetermined logical predicate.
However, the ``difficulty'' of this property is that the logical predicate (usually) verifies a cryptographic proof, which processes might not know a priori (see \Cref{subsection:intution_extended}).

In a nutshell, we make our original formalism more expressive by (1) making the input ($\mathcal{V}_I$) and output ($\mathcal{V}_O$) spaces ``unknown'' to the processes, and (2) taking into account ``proposals'' of faulty processes.
In the rest of the paper:
\begin{compactitem}
    \item We refer to the formalism introduced in the main body of the paper as the ``original formalism''.
    
    \item We refer to the formalism we introduce below as the ``extended formalism''.
\end{compactitem}
We start by giving an intuition behind our extended formalism (\Cref{subsection:intution_extended}).
Then, we introduce some preliminaries (\Cref{subsection:preliminaries_extended}).
Finally, we (incompletely) define our extended formalism (\Cref{subsection:validity_extended}).


\subsection{Intuition} \label{subsection:intution_extended}

In the original formalism, processes know the entire input space $\mathcal{V}_I$ and the entire output space $\mathcal{V}_O$.
That is, processes are able to ``produce'' any value which belongs to $\mathcal{V}_I$ or $\mathcal{V}_O$.
However, this assumption limits the expressiveness of our formalism as it is impossible to describe a Byzantine consensus problem in which input or output spaces are not a priori known.
Let us give an example.

Imagine a committee-based blockchain which establishes two roles:
\begin{compactitem}
    \item \emph{Clients} are the users of the blockchain.
    They issue \emph{signed} transactions to the blockchain.
    
    \item \emph{Servers} are the operating nodes of the blockchain.
    Servers receive signed transactions issued by the clients, and solve the Byzantine consensus problem to agree on the exact order the transactions are processed.
\end{compactitem}
As the servers propose transactions \emph{signed by the clients} and they do not have access to the private keys of the clients, the servers do not know the input space $\mathcal{V}_I$ nor the output space $\mathcal{V}_O$ of the Byzantine consensus problem.
Hence, our original formalism cannot describe the Byzantine consensus problem in the core of the aforementioned blockchain.

\paragraph{Extended vs. original formalism.}
As highlighted above, the main difference between the two formalisms is that the extended one allows us to specify the ``knowledge level'' of the input and output spaces.
In the extended formalism, a process is able to ``learn'' output values by observing input values.
That is, we define a \emph{discovery function} that defines which output values are learned given observed input values.
In the committee-based blockchain example, once a server observes signed (by the issuing clients) transactions $\mathit{tx}_1$ and $\mathit{tx}_2$, it learns the following output values: (1) $\mathit{tx}_1$, (2) $\mathit{tx}_2$, (3) $\mathit{tx}_1 || \mathit{tx}_2$, and (4) $\mathit{tx}_2 || \mathit{tx}_1$.\footnote{We denote by ``$||$'' the concatenation operation.} 

The second difference between the original and the extended formalism is that the extended formalism takes into account ``proposals'' of the faulty processes.
Indeed, the original formalism does not enable us to define which values are admissible given the adversary's knowledge of the input space.
Think of the aforementioned example with a blockchain system.
If no process (correct or faulty) obtains a transaction $\mathit{tx}$, $\mathit{tx}$ cannot be decided.
However, even if \emph{only} a faulty process obtains a transaction $\mathit{tx}$, $\mathit{tx}$ could still be an admissible decision.
This scenario can be described by the extended formalism, and not by the original one.

\subsection{Preliminaries} \label{subsection:preliminaries_extended}

We denote by $\mathcal{V}_I$ the input space of Byzantine consensus.
Similarly, $\mathcal{V}_O$ denotes the output space.

\paragraph{Membership functions.}
We define two \emph{membership functions}:
\begin{compactitem}
    \item $\mathsf{valid\_input} : \{ 0, 1\}^* \to \{ \mathit{true}, \mathit{false} \}$: specifies whether a bit-sequence belongs to the input space $\mathcal{V}_I$.
    
    \item $\mathsf{valid\_output} : \{ 0, 1 \}^* \to \{ \mathit{true}, \mathit{false} \}$: specifies whether a bit-sequence belongs to the output space $\mathcal{V}_O$.
\end{compactitem}
We assume that each process has access to these two functions.
That is, each process can verify whether an arbitrary sequence of bits belongs to the input ($\mathcal{V}_I$) or output ($\mathcal{V}_O$) space.
In the case of a committee-based blockchain (\Cref{subsection:intution_extended}), the membership functions are signature-verification functions.


\paragraph{Discovery function.}
We define a function $\mathsf{discover}$: $2^{\mathcal{V}_I} \to 2^{\mathcal{V}_O}$.
Given a set of proposals $V_I \subseteq \mathcal{V}_I$, $\mathsf{discover}(V_I) \subseteq \mathcal{V}_O$ specifies the set of decisions which are ``discoverable'' by $V_I$.
We assume that each process has access to the $\mathsf{discover}(\cdot)$ function.
Moreover, for any two sets $V_I^1, V_I^2$ with $V_I^1 \subseteq V_I^2$, $\mathsf{discover}(V_I^1) \subseteq \mathsf{discover}(V_I^2)$; in other words, ``knowledge'' of the output space can only be improved upon learning more input values.

Let us take a look at the committee-based blockchain example again (\Cref{subsection:intution_extended}).
If a server obtains a proposal $\mathit{tx}$, it learns $\mathit{tx}$ as a potential decision.
We model this ``deduction'' concept using the $\mathsf{discover}(\cdot)$ function: $\mathsf{discover}\big( \{\mathit{tx}\} \big) = \{\mathit{tx}\}$.


\paragraph{Adversary pool.}
Given an execution $\mathcal{E}$, $\mathcal{P}(\mathcal{E}) \subseteq \mathcal{V}_I$ defines the \emph{adversary pool} in $\mathcal{E}$.
Informally, the adversary pool represents the input values the adversary ``knows''.
In the example of a committee-based blockchain (\Cref{subsection:intution_extended}), the adversary pool is a set of signed transactions which the adversary ``learns'' from the clients.

We underline that the adversary pool is an abstract concept.
Specifically, the adversary pool represents the ``starting knowledge'' the adversary has.
However, the notion of the ``starting knowledge'' must be precisely defined once all particularities of the exact considered system are taken into account.
Due to sophisticated details (such as the aforementioned one), we believe that a formalism suitable for blockchain-specific validity properties deserves its own standalone paper.

\subsection{Validity} \label{subsection:validity_extended}

We start by restating the definition of process-proposal pairs.
A \emph{process-proposal} pair is a pair $(P, v)$, where (1) $P \in \allprocesses$ is a process, and (2) $v \in \mathcal{V}_I$ is a proposal.
Given a process-proposal pair $\mathit{pp} = (P, v)$, $\mathsf{proposal}(\mathit{pp}) = v$ denotes the proposal associated with $\mathit{pp}$.

An \emph{input configuration} is a tuple $\big[ \mathit{pp}_1, \mathit{pp}_2, ..., \mathit{pp}_x, \rho \big]$ of $x$ process-proposal pairs and a set $\rho \subseteq \mathcal{V}_I$, where (1) $n - t \leq x \leq n$, (2) every process-proposal pair is associated with a distinct process, and (3) if $x = n$, $\rho = \emptyset$.
Intuitively, an input configuration represents an assignment of proposals to correct processes, as well as a ``part'' of the input space known to the adversary.
For example, an input configuration $\big[ (P_1, v), (P_2, v), (P_3, v), \{v, v', v''\} \big]$ describes an execution in which (1) only processes $P_1$, $P_2$, and $P_3$ are correct, (2) processes $P_1$, $P_2$, and $P_3$ propose the same value $v$, and (3) faulty processes know only $v$, $v'$, and $v''$.

We denote by $\mathcal{I}$ the set of all input configurations.
For every input configuration $c \in \mathcal{I}$, we denote by $c[i]$ the process-proposal pair associated with process $P_i$; if such a process-proposal pair does not exist, $c[i] = \bot$.
Moreover, we define by $\mathsf{pool}(c)$ the set of input values associated with $c$ (the ``$\rho$'' field of $c$).
Next, $\process{c} = \{ P_i \in \Pi \,|\, c[i] \neq \bot\}$ denotes the set of all processes included in $c$.
Finally, $\mathsf{correct\_proposals}(c) = \{ v \in \mathcal{V}_I \,|\, \exists i \in [1, n]: c[i] \neq \bot \land \mathsf{proposal}(c[i]) = v \}$ denotes the set of all proposals of correct processes (as specified by $c$).

Given (1) an execution $\mathcal{E} \in \mathit{execs}(\mathcal{A})$, where $\mathcal{A}$ is an algorithm with the $\mathsf{propose}(\cdot)/\mathsf{decide}(\cdot)$ interface, and (2) an input configuration $c \in \mathcal{I}$, we say that $\mathcal{E}$ \emph{corresponds} to $c$ ($\mathsf{input\_conf}(\mathcal{E}) = c$) if and only if (1) $\process{c} = \mathit{Corr}_{\mathcal{A}}(\mathcal{E})$, (2) for every process $P_i \in \mathit{Corr}_{\mathcal{A}}(\mathcal{E})$, $P_i$'s proposal in $\mathcal{E}$ is $\mathsf{proposal}(c[i])$, and (3) $\mathcal{P}(\mathcal{E}) = \mathsf{pool}(c)$.

A validity property $\mathit{val}$ is a function $\mathit{val}: \allconfigurations \to 2^{\mathcal{V}_O}$ such that, for every input configuration $c \in \allconfigurations$, $\mathit{val}(\mathit{c}) \neq \emptyset$.
Algorithm $\mathcal{A}$, where $\mathcal{A}$ exposes the $\mathsf{propose}(\cdot)/\mathsf{decide}(\cdot)$ interface, \emph{satisfies} a validity property $\mathit{val}$ if and only if, in any execution $\mathcal{E} \in \mathit{execs}(\mathcal{A})$, no correct process decides a value $v' \notin \mathit{val}\big( \mathsf{input\_conf}(\mathcal{E}) \big)$.
That is, an algorithm satisfies a validity property if and only if correct processes decide only admissible values.

    
    

\paragraph{Assumptions on executions.}
Lastly, we introduce two assumptions that conclude our proposal for the extended formalism.

\begin{assumption} \label{assumption:all_executions}
For every execution $\mathcal{E}$ of any algorithm $\mathcal{A}$ which solves the Byzantine consensus problem with some validity property, if a correct process $P$ decides a value $v' \in \mathcal{V}_O$ in $\mathcal{E}$, then $v' \in \mathsf{discover}\big( \mathsf{correct\_proposals}(c) \cup \mathsf{pool}(c) \big)$, where $\mathsf{input\_conf}(\mathcal{E}) = c$.
\end{assumption}

\Cref{assumption:all_executions} states that correct processes can only decide values which are ``discoverable'' using all the proposals of correct processes and the knowledge of the adversary.
For example, if every correct process proposes the same value $v \in \mathcal{V}_I$ and the adversary pool contains only $v' \in \mathcal{V}_I$, then a correct process can only decide a value from $\mathsf{discover}(\{v, v'\})$.

Next, we introduce an assumption concerned only with the canonical executions (executions in which faulty processes do not take any computational step).

\begin{assumption} \label{assumption:extended}
For every canonical execution $\mathcal{E}$ of any algorithm $\mathcal{A}$ which solves the Byzantine consensus problem with some validity property, if a correct process $P$ decides a value $v' \in \mathcal{V}_O$ in $\mathcal{E}$, then $v' \in \mathsf{discover}\big( \mathsf{correct\_proposals}(c) \big)$, where $\mathsf{input\_conf}(\mathcal{E}) = c$.
\end{assumption}

Intuitively, \Cref{assumption:extended} states that, if faulty processes are silent, correct processes can only decide values which can be discovered using their own proposals.
In other words, correct processes cannot use ``hidden'' proposals (possessed by the silent adversary) to discover a decision.

Finally, we underline that these two assumptions do not completely prevent ``unreasonable'' executions.
For example, given these two assumptions, a (correct or faulty) process is still able to send a message with a value which cannot be discovered using the proposals of correct processes and the adversary pool.
Hence, an assumption that prevents such an execution should be introduced.
Thus, due to the complexity of the extended formalism, we leave it out of this paper.
In the future, we will focus on this interesting and important problem.

\end{document}